\documentclass{cmslatex}
\usepackage[paperwidth=7in, paperheight=10in, margin=.875in]{geometry}

\usepackage{amssymb}
\usepackage{amsmath}
\usepackage{my_maccros}
\usepackage{graphics,float}
\usepackage{mathtools}
\usepackage[colorlinks=true]{hyperref}

\renewcommand{\theequation}{\arabic{section}.\arabic{equation}}

\newtheorem{thm}{Theorem}[section]

 \newtheorem{rem}[thm]{Remark}

\begin{document}
 \title{Nonlinear wave-current interactions in shallow water}
\author{David Lannes\thanks{IMB, Universit\'e de Bordeaux, France,  David.Lannes@math.u-bordeaux.fr}
   \and{Fabien Marche\thanks{IMAG, Universit\'e de Montpellier, and Inria team LEM0N, Montpellier, France,  Fabien.Marche@math.univ-montp2.fr}}}

\pagestyle{myheadings} \markboth{Nonlinear wave currents interactions in shallow water}{D. Lannes and F. Marche} \maketitle

\begin{abstract}
          We study here the propagation of long waves in the presence of vorticity. In the irrotational framework, the Green-Naghdi equations (also called Serre or fully nonlinear Boussinesq equations) are the standard model for the propagation of such waves. These equations couple the surface elevation to the vertically averaged horizontal velocity and are therefore independent of the vertical variable. In the presence of vorticity, the dependence on the vertical variable cannot be removed from the vorticity equation but it was however shown in \cite{castroLannes:2014aa} that the motion of the waves could be described using an extended Green-Naghdi system. In this paper we propose an analysis of these equations, and show that they can be used to get some new insight into wave-current interactions. We show in particular that solitary waves may have a drastically different behavior in the presence of vorticity and show the existence of solitary waves of maximal amplitude with a peak at their crest, whose angle depends on the vorticity. We also show some simple numerical validations. Finally, we give some examples of wave-current interactions with a non trivial vorticity field and topography effects.
\end{abstract}


\section{Introduction}

\subsection{General setting}
Several models have been derived for the description of nearshore dynamics. One of the most widely spread is certainly the Nonlinear Shallow Water (NSW) model which is a nonlinear hyperbolic system coupling the time evolution of the surface elevation $\zeta$ to the vertically averaged horizontal component of the velocity $\ovv$. This system is derived from the free surface Euler equations by averaging in the vertical direction and neglecting all the terms of order $O(\mu)$, where the {\it shallowness parameter} $\mu$ is defined as
$$
\mu=\frac{H_0^2}{L^2}=\frac{(\mbox{typical
    depth})^2}{(\mbox{horizontal length scale})^2}.
$$
The NSW equations are  however not fully satisfactory since they neglect all the dispersive effects that play a very
important role in many situations, and in particular during the shoaling phase. These dispersive terms are of order $O(\mu)$ and are therefore neglected by the NSW equations. Keeping them in the equations, and neglecting only the $O(\mu^2)$ terms, one obtains a more accurate -- but mathematically and numerically more complicated -- set of equations known as the Serre \cite{Serre,SuGardner}, or Green-Naghdi \cite{GreenNaghdi,KBEW}, or fully nonlinear Boussinesq \cite{WKGS} equations. We shall refer to these models here as the Green-Naghdi (GN) equations. Contrary to the weakly nonlinear Boussinesq models that go back to Boussinesq himself, no smallness assumption is made on the size of the surface perturbations. We refer to \cite{lannes:book} for a rigorous derivation and a mathematical justification  (in the sense that their solutions remain close to the exact solution of the free surface Euler equations) of all these models. 
If the numerical approximation of various Boussinesq-type equations has attracted a lot of attention for the last 20 years (see for instance among the recent studies \cite{kalisch,  erduran:2007, eskilsson:2006,  delis:2012, engsig:2008, mario:2013, roeber:2012, shiach, soares}), it is mostly recently that discrete formulations for the GN equations have been proposed. Denoting by $d$ the horizontal dimension, we can refer for instance, in the case $d=1$, to  \cite{rchgth89, Bonneton20111479, MR2811693} for hybrid Finite-Volume (FV) and Finite-Difference (FD) discretizations, \cite{DuranMarche:2014ab, LGLX, pandi:panda} for discontinuous-Galerkin (dG) formulations, \cite{cie2007} for a compact FV approach, \cite{Mitso} for a Finite-Element (FE) approach on flat bottom or \cite{mario:2015} for an hybrid FV-FE formulation. There is even less studies in the case $d=2$, see \cite{lannes_marche:2014, LGH, SKHGG}.
These equations have also been adapted to handle wave breaking by adding an artificial viscous term to the momentum equation (see for instance \cite{CKDKC,CBB, KCKD, SM}) or by locally switching to the NSW equations in the vicinity of broken waves and using shock capturing schemes \cite{delis:2014, tissier2, tonelli:2010}. We refer to \cite{rchgth89} for a recent review on these aspects.

\begin{figure}
\centering
\includegraphics[width=0.55\textwidth]{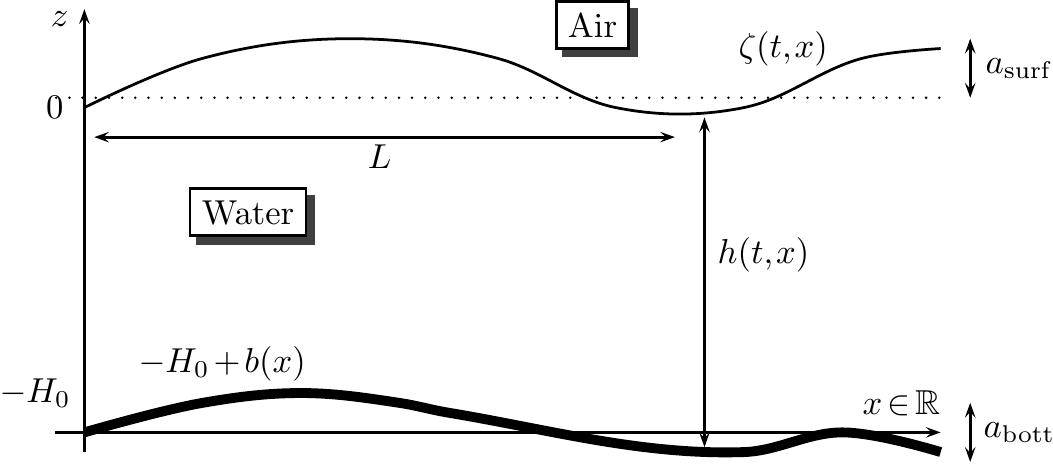}
\caption{Main notations.}\label{fignota}
\end{figure}


The Green-Naghdi and more generally most of the Boussinesq-type models rely on the assumption that the flow is irrotational or almost irrotational. Such an assumption is satisfied in most configurations but may fail in the surf zone where wave breaking can create vorticity currents (rip currents) or in present of some underlying current. The difficulty to describe wave motion in the presence of vorticity is that the dynamics of the flow is in general genuinely $(d+1)$-dimensional while in the irrotational framework the dynamics is only $d$-dimensional (vertical averaging has been used to remove the vertical variable). 

It is shown in \cite{CKDST} that the GN equations can describe rotational flows with purely vertical vorticity. In horizontal dimension $d=1$, it is shown in \cite{MS} that vorticity is responsible for the presence of an additional term in the momentum equation, which is coupled to the standard vorticity equation (see also \cite{KLS,KL} for a related approach).

Following the original approach by Green and Naghdi \cite{GreenNaghdi}, several authors  \cite{EWW,SW,ZDE} assumed a polynomial structure of the velocity profile and solved the mass and momentum equations projected on such a basis of functions. This approach is compatible with the presence of vorticity. An interesting and recent refinement for a better treatment of the surf zone consists in coupling this approach with shallow water asymptotics  \cite{ZKNDW}.

We follow in this paper another approach initially developed in \cite{castroLannes:2014aa} where it is shown that additional terms are necessary in the momentum equation in the presence of vorticity. Contrary to other approaches, these additional terms are determined through the resolution of $d$-dimensional evolution equations and do not require the resolution of the $(d+1)$-dimensional vorticity equation. The procedure is reminiscent of turbulence theory with the difference that no artificial closure is needed here: based on the controls on the solutions of the full Euler equations established in \cite{castro:lannes:2014:1}, one can show that the cascade of equations is finite at the precision of the model (we also point out the related work \cite{RG} where a physical modeling of the closure is used instead to handle turbulent bores). The resulting equations are an extended GN system with additional advection-like equations for the vorticity related terms. \\

This paper is concerned with the horizontal one dimensional case $d=1$. Its first goal  is to show that this approach can be used to get some new insight into wave-current interactions; we show for instance that the behavior of solitary waves can be drastically different in the presence of vorticity, leading to extremal peaked solitary waves with an angle at the crest that depends on the vorticity. The second goal of this paper is to propose a simple numerical scheme to numerically solve these extended GN equations and to highlight that despite the fact that these equations of motion are purely $d$-dimensional, they can be used to reconstruct the internal velocity field, even in the presence of non trivial vorticity and topography.

\subsection{The models}

In the case where the horizontal dimension $d$ is equal to one and following \cite{Bonneton20111479}, the irrotational GN 
equations can be formulated as
\be\label{eq6imphhVbis}
	\left\lbrace
	\begin{array}{l}
	\dsp \dt h + \dx \cdot (h \ovv)=0,\\
       	\dsp (I+ \mathbb{T})[\dt (h\ovv) 
        +\dx (h\ovv^2)] + gh\dx\zeta 
	+h\cQ_1(\ovv)  =0,
	\end{array}\right.
\ee
where we recall that $\zeta$ is the elevation of the wave with respect
to the rest level and that $\ovv$ is the vertically averaged
horizontal velocity, while $g$ stands for the acceleration of
gravity and $h$ is the total water height
$$
	h=H_0+\zeta-b,
$$
where $\{z=H_0-b(x)\}$ is a parametrization of the bottom (see Figure
\ref{fignota}). Finally, the linear operator
$\mathbb{T}=\mathbb{T}[h,b]$ and the quadratic form $\cQ_1(\cdot)=\cQ_1[h,b](\cdot)$ 
are defined by
\bea
	\label{eq2}
\mathbb{T}[h,b]w&=&h\cT[h,b](\frac{1}{h}w),\\
		\cT[h,b]w&=&
	-\frac{1}{3h}\dx(h^3 \dx w)+
\frac{1}{2h}\big( \dx(h^2 w\dx b)
	-h^2\dx w\dx b\big)
	+w(\dx b)^2,\\
	\label{eq3}
	\cQ_1[h,b](v)&=&\frac{2}{3h}\dx(h^3 (\dx v)^2)
	+  h (\dx v)^2\dx b
	+\frac{1}{2h}\dx(h^2 v^2 \dx^2 b)
	+v^2 \dx^2 b\dx b.
\eea
This formulation does not require the computation of any third-order
derivative, allowing for more robust numerical computations,
especially when the waves become steeper. Note also that if one removes the operator $\mathbb{T}$ and the nonlinearity $\cQ_1$ from the second equation in \eqref{eq6imphhVbis}, the model reduces to the standard NSW equations; these two terms accounts therefore for the $O(\mu)$ dispersive and nonlinear terms specific to the GN equations.

One of the main features of the Green-Naghdi model is that it allows
the description of  $d+1$ dimensional waves ($d$ being the horizontal
dimension) by a set of $d$-dimensional equations (independent on the vertical
variable $z$), hereby leading to considerable gains in mathematical
simplicity and computational time. The $d$-dimensional nature of the
flow is due to the fact that the flow is assumed to be {\it
  irrotational}; indeed, the velocity field ${\bU}$ in the fluid
domain then derives from a scalar velocity potential $\Phi$
(i.e. $\bU=\nabla_{X,z}\Phi$) and as remarked by Zakharov
\cite{Zakharov:1968aa} and Craig-Sulem \cite{CS} the free
surface $(d+1)$-dimensional Euler equations can then be reduced to an
Hamiltonian system coupling the surface elevation $\zeta$ to $\psi$, the trace
at the surface of the velocity potential. Both $\zeta$ and $\psi$
depend only on time and on the ($d$-dimensional) horizontal
variable $X$. The Green-Naghdi equation being obtained by an asymptotic
expansion in terms of the shallowness parameter $\mu$ of the free
surface Euler equations (see \cite{AlvarezSamaniego:2008p227,lannes:book} for a full mathematical justification
of this approximation), it is no surprise that they are also $d$
dimensional.

In presence of vorticity, the situation is drastically different since
the dynamics of the vorticity $\omega=\curl \bU$ is in general fully
$(d+1)$-dimensional. The Zakharov-Craig-Sulem formulation has recently
been generalized in \cite{castro:lannes:2014:1} to the rotational case; this
generalization, also formally hamiltonian, couples the evolution of
$\zeta$ and $\psi$ as in the irrotational case\footnote{Note however
  that in presence of vorticity, the velocity field $\bU$ does not
  derive from a scalar potential, and that an alternative definition
  is needed for $\psi$. Namely, it is defined such that $\nabla\psi$
  is the projection onto (horizontal) gradient vector field of the
  horizontal component of the tangential velocity at the surface.}, but this evolution is
now also coupled to the evolution of the vorticity field which depends in
general on all the space variables. One should therefore expect that
generalizations of the Green-Naghdi equations in presence of vorticity
have a full $(d+1)$-dimensional dependence in the space variables,
hereby implying a considerable increase of computational time. It has been shown recently in
\cite{castroLannes:2014aa} that this is not the case. In the case of a
constant vorticity, that is, when
$$
\curl \bU=(0,\omega,0)^T\quad \mbox{ with }\quad \omega(t,x,z)=\omega_0=\mbox{cst},
$$
this is not surprising because there is no $z$ dependence coming from the
equation on the vorticity. The vorticity field however induces a shear
which, together with the dispersive effects, make the horizontal
velocity depart from its vertical average. Because of this effect, the
Green-Naghdi equations \eqref{eq6imphhVbis} must be replaced\footnote{The model \eqref{gn:const} actually holds under a rather weak smallness assumption on the topography. More precisely, it is assumed that $\eps\beta\mu^{3/2}=O(\mu^2)$, where 
$$
\eps=\frac{a_{\rm surf}}{H_0}=\frac{\mbox{Amplitude of the waves}}{\mbox{typical depth}},\qquad
\beta=\frac{a_{\rm bott}}{H_0}=\frac{\mbox{Amplitude of the bottom variations}}{\mbox{typical depth}}.
$$
In full generality, an additional topography term is needed, see (1.27) in \cite{castroLannes:2014aa}.} by
\be\label{gn:const}
\begin{cases}
\dt h + \dx(h\bar{v}) = 0,\\
(1+\mathbb{T})[\dt(h\bv) + \dx(h\bv^2)] + gh\dx \zeta+ h\cQ_1(\bv) \\
\phantom{xxxxxxxxxxx\dx(h\bv^2)] + }+ \dx(\frac{1}{12}h^3\omega_0^2) + h\cC(\omega_0 h,\bv)  =0,
\end{cases}
\ee
with $\cC(\omega_0 h,\bv)$  obtained by
taking $v^\sharp=\omega_0 h$ in the following expression
\begin{align}
\cC(\dv,\bv) = -\frac{1}{6h}\dx\left( 2h^3\dv\dx^2\bv +\dx(h^3\dv)\dx\bv \right). \label{gen:c1}
\end{align}
For the case of a general vorticity, that is, when (in horizontal dimension $d=1$),
$$
\curl \bU=(0,\omega,0)^T\quad \mbox{ with }\quad \omega(t,x,z)=\dz
u-\dx w
$$
(and $\bU=(u,0,w)^T$), the vorticity $\omega$ satisfies the transport
equation
\begin{equation}\label{eqvort}
\dt \omega +(u\dx +w\dz )\omega=0,
\end{equation}
in which the $z$ dependence cannot be removed. The fact that one can
however derive $z$-independent Green-Naghdi type models in this
framework is therefore more surprising.  The one-dimensional GN system in the
presence of a general vorticity is then given\footnote{The same smallness assumption as for \eqref{gn:const} is made on the topography.} by
\be\label{gn:vort:dim}
\begin{cases}
&\dt h + \dx(h\bar{v}) = 0,\\
&(1+\mathbb{T})[\dt(h\bv) + \dx(h\bv^2)] + gh\dx \zeta + h\cQ_1(\bv)
+\dx E+ h\cC(\dv,\bv) =0,\\
&\dt \dv + \bv\dx\dv + \dv\dx\bv =0,\\
&\dt E + \bv\dx E + 3 E\dx \bv+ \dx F = 0,\\
&\dt F + \bv\dx F + 4 F\dx\bv = 0.
\end{cases}
\ee
Let us briefly comment on this model (we refer to \cite{castroLannes:2014aa} for more details). Due to the presence of the vorticity, there is a vertical dependence of the velocity field inside the fluid domain that may interact nonlinearly with the vertical dependence coming from the dispersive terms. After vertical averaging, this interaction is responsible for the term $h\cC(\dv,\bv)$  in the momentum equation; it is given by \eqref{gen:c1}. The difference with the case of a constant
vorticity is that $v^\sharp$ is now defined as a second order momentum
of the vorticity induced shear velocity,
\begin{equation}\label{defvsharp}
v^\sharp=\frac{12}{h^3}\int_{-H_0+b}^\zeta (z+H_0-b)^2v_{\rm sh}^*
\quad\mbox{ with }\quad
v^*_{\rm sh}=-\int_z^\zeta \omega+\frac{1}{h}\int_{-H_0+b}^\zeta \omega.
\end{equation}
Even though $v^\sharp$ is defined in terms of $\omega$, it is not
necessary to solve the $(1+1)$-dimensional vorticity equation
\eqref{eqvort} to compute it; indeed, it is shown in \cite{castroLannes:2014aa} that it
can be determined from its initial value by solving the third equation
of \eqref{gn:vort:dim}.\\
Similarly, the term $\dx(\frac{1}{12}h^3\omega_0^2)$ that appears in
\eqref{gn:const} is now replaced by $\dx E$, where $E$ is a second
order tensor (represented by a function here since we are in dimension $1$) describing the self quadratic interaction of the
vorticity induced shear inside the fluid domain,
\begin{equation}\label{defE}
E=\int_{-H_0+b}^\zeta (v^*_{\rm sh})^2.
\end{equation}
Here again, one wants to be able to compute $E$ without appealing to
the vorticity equation \eqref{eqvort}. The strategy adopted in
\cite{castroLannes:2014aa} is inspired by an analogy with turbulence
theory and recent works on roll waves and hydraulic jumps \cite{RG1,RG2}. The tensor $E$ is viewed as a ``Reynolds'' tensor where the
``averaging'' is in the present case the vertical integration. Looking for an
equation on $E$ one obtains a cascade of
equations involving tensors of increasing order; but unlike turbulence
theory, there is no need for an artificial closure of the cascade.
Indeed it can be proved that the contribution of the fourth order and
higher tensors are below the overall $O(\mu^2)$ precision of the model
and can therefore be neglected. The last two equations in
\eqref{gn:vort:dim} furnish this finite cascade of equations on the
second order tensor $E$ and the third order tensor $F$ (here again represented by a function in dimension $1$) defined as
\begin{equation}\label{defF}
F=\int_{-H_0+b}^\zeta (v_{\rm sh}^*)^3.
\end{equation}
The generalization of \eqref{gn:vort:dim} to two-dimensional surfaces
is also given in \cite{castroLannes:2014aa} but our focus is here on the
analysis of some properties of \eqref{gn:vort:dim} as well as the
development of a numerical code to compute its solutions.

\vspace{0.2cm}
\begin{rem}
As explained above, the model \eqref{gn:vort:dim} is precise up to
$O(\mu^2)$ terms; lowering the precision to $O(\mu^{3/2})$, one can
work with the simpler model
\be\label{gn:vort:dimsimpl}
\begin{cases}
&\dt h + \dx(h\bar{v}) = 0,\\
&(1+\mathbb{T})[\dt(h\bv) + \dx(h\bv^2)] + gh\dx \zeta + h\cQ_1(\bv)
+\dx E=0,\\
&\dt E + \bv\dx E + 3 E\dx \bv= 0.
\end{cases}
\ee
\end{rem}
\begin{rem}
Keeping the precision $O(\mu^2)$, a simplified model can also be obtained
in the situation where $F$ is initially almost equal to zero (this is
the case when the vorticity is constant or in the situation considered
in \S \ref{sect2Dex1} for instance). Removing $F$ from
\eqref{gn:vort:dim} one then obtains the reduced model
\be\label{modeleFzero}
\begin{cases}
&\dt h + \dx(h\bar{v}) = 0,\\
&(1+\mathbb{T})[\dt(h\bv) + \dx(h\bv^2)] + gh\dx \zeta + h\cQ_1(\bv)
+\dx E+ h\cC(\bv,\dv) =0,\\
&\dt \dv + \bv\dx\dv + \dv\dx\bv =0,\\
&\dt E + \bv\dx E + 3 E\dx \bv = 0.
\end{cases}
\ee
\end{rem}

\subsection{Organization of the paper}

In Section \ref{section:sol_wave}, we study the existence of solitary waves for the GN system with vorticity \eqref{gn:vort:dim}.
We show in \S \ref{sect:derSol} that the existence of  smooth solitary waves can be reformulated as an ODE problem. The existence of solutions is then established in \S \ref{ode:analysis} where we also comment on the qualitative differences with the irrotational case. For instance, while there are solitary waves of arbitrary amplitude for the standard GN equations, there are configurations with non trivial vorticities for which solitary waves cannot exceed a critical amplitude. Solitary waves of critical amplitude are then studied in \S \ref{ode:analysis:peak} where we show that these extremal solitary waves have a peak at their crest, whose angle depends on the vorticity.

We then present in Section \ref{sect:num:scheme} the numerical scheme we propose to solve  \eqref{gn:vort:dim}. After a simple renormalisation of the system using the mass conservation equation, we present in \S \ref{sect:splitting} a simple splitting scheme inspired by previous works on the standard GN equations. This splitting involves a conservative propagation step and a dispersive correction step. The conservative step is studied in \S \ref{sect:conservative}; as in the irrotational case, it is of hyperbolic type but because of the extra unknowns due to the vorticity, its structure is more complicated. In particular, there are now three wave speeds (instead of two in the irrotational case). A corresponding finite volume scheme is proposed, for which higher order extensions are constructed. The study of the dispersive step being similar to the irrotational case, we just briefly recall the main points in \S \ref{sect:S2}.

Section \ref{validation} is then devoted to the numerical validation of this scheme. The different kinds of smooth solitary waves predicted in Section \ref{section:sol_wave} are numerically observed in \S \ref{sect:SWP} and used to evaluate the convergence rate. We also observe numerically in \S \ref{sect:PSW} the existence of  the extremal peaked solitary waves exhibited in \S \ref{ode:analysis:peak}. We provide in \S \ref{sect:vortshoal} a numerical simulation involving a non-flat topography; we use this example to show that the vorticity may play a considerable role on the shoaling of waves.  

Finally, we detail in Section \ref{sec:velocity} how the system of equations \eqref{gn:vort:dim} can be
used to describe the dynamics of the $(d+1)$ velocity field $\bU=(u,w)^T$ at any time, up to a $O(\mu^\frac{3}{2})$ accuracy, and show how the previous discrete formulation of Section \ref{sect:num:scheme} may be easily modified to perform the corresponding velocity reconstructions. This process is illustrated by two prospective examples of wave-current interactions, involving a non trivial vorticity field and topography effects.

\section{Solitary waves}\label{section:sol_wave}

We investigate here the existence of solitary waves for the
Green-Naghdi equations with vorticity \eqref{gn:vort:dim}. We show in
\S \ref{sect:derSol} that smooth solitary waves must satisfy a second
order ODE. This ODE is solved in \S \ref{sect:derSol}, while extremal peaked solutions are studied in  \S \ref{ode:analysis:peak}.
\subsection{Derivation of the ODE for the shape of the solitary waves} \label{sect:derSol}

Our purpose here is to show that smooth solitary waves, if they exist,
must satisfy a second order ODE. We consider here flat bottoms
(i.e. $b=0$), and the system (\ref{gn:vort:dim}) can therefore be written
\begin{equation}\label{GNvort1dfin}
\begin{cases}
\zeta_t+(h\ovv)_x=0,\\
\dsp \ovv_t+g\dx\zeta
+\ovv\dx\ovv +\frac{1}{h} E_x-\frac{1}{6h}
\big[2h^3 v^\sharp \ovv_{xx}+ (h^3 v^\sharp)_x\ovv_x\big]_x=\frac{1}{3}\frac{1}{h}\big[h^3(\ovv_{xt}+
\ovv\ovv_{xx}- \ovv_x^2)\big]_x 
\\
\dsp  v^\sharp_t+ (\ovv v^\sharp)_x =0,\\
\dsp  \big(\frac{E}{h^3}\big)_t+\ovv  \big(\frac{E}{h^3}\big)_x+ \frac{1}{h^3}F_x=0,\vspace{1mm}\\
\dsp \big(\frac{F}{h^4}\big)_t+\ovv\big(\frac{F}{h^4}\big)_x =0,
\end{cases}
\end{equation}
with $h=H_0+\zeta$.
We look for solitary waves solutions to \eqref{GNvort1dfin},
i.e. solutions of the form
$$
(\zeta,\ovv,v^\sharp,E,F)(t,x)=(\underline{\zeta},\underline{\ovv},\underline{v}^\sharp,\uE,\uF)(x-ct),
$$
for some constant $c\in \R$, and with $\underline{\zeta}$ and
$\underline{v}$ vanishing at infinity, over a current that might not
vanish at infinity, that is, we assume that
$$
\lim_{\pm\infty} (\underline{\zeta},\underline{v})=0
\quad\mbox{ and }\quad \lim_{\pm\infty}(\underline{v}^\sharp,\underline{E},\underline{F})=(v_\infty^\sharp,
E_\infty,F_\infty)
$$
for some constants $E_\infty$, $F_\infty$, $v_\infty^\sharp$.
Such solutions should satisfy
(for the sake of clarity, we do not underline the functions in the expressions below)
\begin{equation}\label{GNvort1dfin2}
\begin{cases}
[-(c-\ovv) h]_x=0,\\
\dsp -(c-\ovv) \ovv_x+g\dx\zeta
 +\frac{E_x}{h} -\frac{1}{6h}
\big[2h^3 v^\sharp \ovv_{xx}+ (h^3 v^\sharp)_x\ovv_x\big]_x=-\frac{1}{3h}\big[h^3((c-\ovv)\ovv_{xx}+ \ovv_x^2)\big]_x \vspace{1mm}
\\
\dsp  \big[(c-\ovv) v^\sharp\big]_x =0,\vspace{1mm}\\
\dsp  -(c-\ovv)  \big(\frac{E}{h^3}\big)_x+ \frac{1}{h^3}F_x=0,\vspace{1mm}\\
\dsp (c-\ovv)\big(\frac{F}{h^4}\big)_x =0.
\end{cases}
\end{equation}
Integrating the first equation, and using the fact that $\zeta$ and
$\ovv$ vanish at infinity, one readily deduces
\begin{equation}\label{ann1}
(c-\ovv) h=cH_0.
\end{equation}
Multiplying the second equation by $h$ and integrating in $x$, we
therefore get
\begin{align}
-cH_0\ovv+\frac{g}{2}(h^2-H_0^2)
 +  (E-E_\infty)- &\frac{1}{6}
\big(2h^3 v^\sharp \ovv_{xx} + (h^3 v^\sharp)_x\ovv_x\big)\label{ann2}\\
&=-\frac{1}{3}\big(h^2 cH_0\ovv_{xx}+ h^3\ovv_x^2\big),\nonumber
\end{align}
and we need to determine $v^\sharp$ and $E$. Let us proceed first with $v^\sharp$.
From the third equation in \eqref{GNvort1dfin2}, we have
$$
(c-\ovv)v^\sharp=cv^\sharp_\infty,
$$
which, together with \eqref{ann1}, yields
\begin{equation}\label{expvs}
v^\sharp=\frac{h}{H_0} v^\sharp_\infty.
\end{equation}
We now turn to derive an expression for $E$. From the last equation, we get that
\begin{equation}\label{expF}
F=\frac{h^4}{H_0^4} F_\infty.
\end{equation}
Together with \eqref{ann1}, this allows one to rewrite
the fourth equation as
$$
-c H_0\big(\frac{E}{h^3}\big)_x+4  \frac{F_\infty}{H_0^4} h h_x=0,
$$
and therefore
\begin{equation}\label{expE}
E=E_\infty+(\frac{h^3}{H_0^3}-1)E_\infty+2\frac{F_\infty}{c}\frac{(h^2-H_0^2)h^3}{H_0^5}.
\end{equation}
Plugging \eqref{expvs} and \eqref{expE} into \eqref{ann2} we obtain
\begin{align*}
-cH_0\ovv+\frac{g}{2}(h^2-H_0^2)
 + (\frac{h^3}{H_0^3}-1)E_\infty  &+2\frac{F_\infty}{c}\frac{(h^2-H_0^2)h^3}{H_0^5}
 -\frac{h^2}{3H_0}v^\sharp_\infty
[ h^2\ovv_{x}]_x \\
&=-\frac{1}{3}\big(h^2 cH_0\ovv_{xx}+ h^3\ovv_x^2\big). 
\end{align*}
Since \eqref{ann1} implies that $ h^2 v_x=cH_0h_x$ and
$v=c\frac{h-H_0}{ h}$, we deduce further
that
\begin{align*}
\nonumber
-c^2H_0\frac{h-H_0}{h}+\frac{g}{2}(h^2-H_0^2)
 + (\frac{h^3}{H_0^3}-1)E_\infty &+2\frac{F_\infty}{c}\frac{(h^2-H_0^2)h^3}{H_0^5}  -\frac{1}{3} cv^\sharp_\infty
h^2h_{xx} 
\\ &=-\frac{c^2}{3} H_0^2 h\big[\frac{1}{h}h_x\big]_x.
\end{align*}
This leads us to the following definition of a (smooth) solitary wave:

\vspace{0.2cm}
\begin{definition}\label{defSolWav}
A solitary wave of speed $c$ for \eqref{GNvort1dfin} is a mapping
$$
(t,x)\in \R^2\mapsto (\underline{\zeta},\underline{\ovv}, \underline{v}^\sharp,
\underline{E}, \underline{F})(x-ct)
$$
 such that there exists $h\in C^2(\R)$ and
$E_\infty>0$, $v^\sharp_\infty\in \R$ and $F_\infty\in \R$ such that
\begin{align}
&\underline{\zeta}=h-H_0,\quad \underline{\ovv}=c\frac{h-H_0}{h},\quad
\underline{v}^\sharp=\frac{h}{H_0}v^\sharp_\infty,\\
&\underline{E}=\frac{h^3}{H_0^3}E_\infty+2\frac{F_\infty}{c}\frac{(h^2-H_0^2)h^3}{H_0^5},\quad
\underline{F}=\frac{h^4}{H_0^4} F_\infty
\end{align}
and $h$ solves the ODE
\begin{align}\label{ODEbase}
\frac{1}{3}c(cH_0^2-v_\infty^\sharp h^2)h_{xx}=\frac{h-H_0}{2 h}(2c^2H_0- & gh(h+H_0)) -  (\frac{h^3}{H_0^3}-1)E_\infty  \\&-2\frac{F_\infty}{c}\frac{(h^2-H_0^2)h^3}{H_0^5}+\frac{c^2}{3}H_0^2\frac{h_x^2}{h}\nonumber
\end{align}
on $\R$ and satisfies $\lim_{\pm \infty} h=H_0$. The function $h$ is called
the \emph{profile} of the solitary wave.
\end{definition}

\subsection{Existence of smooth solitary waves}\label{ode:analysis}

We first consider here the case where $F_\infty=0$ and prove the
existence of solitary waves in the sense of Definition \ref{defSolWav}.

\vspace{0.2cm}
\begin{proposition}\label{propSolWav}
Let $E_\infty>0$, $v^\sharp_\infty\in \R$ and $F_\infty=0$. Let also
$h_{\rm max}>H_0$. \\
{\bf i.} Up to translations, there can be at most two solitary waves of maximal height
$h_{\rm max}$ for \eqref{GNvort1dfin}; if they exist, they have
opposite speed $\pm \underline{c}$, with
$$
\underline{c}=\Big(gh_{\rm max}+\frac{h_{\rm max}(h_{\rm max}+2H_0)}{H_0^3}E_\infty\Big)^{1/2}.
$$
{\bf ii.} The solitary wave of speed $\underline{c}$
(resp. $-\underline{c}$) exists if and only if the following condition holds 
$$
\underline{c}(\underline{c}H_0^2-v_\infty^\sharp h_{\rm max}^2)>0 \qquad
(\mbox{resp. } -\underline{c}(-\underline{c}H_0^2-v_\infty^\sharp
h_{\rm max}^2)>0 ).
$$
The profile of the solitary wave then attains its maximal
value at a unique point $x_{\rm max}$ and it is symmetric with respect
to the axis $x=x_{\rm max}$ and decaying on the half-line $x>x_{\rm max}$.
\end{proposition}

\vspace{0.2cm}
\begin{proof}
{\bf Step 1.} We derive here an expression for $h_x^2$ in terms of
$h$. For later investigations, we deal with the general case $F_\infty\in \R$ here.
Multiplying the differential equation \eqref{ODEbase} by $h_x$ and
dividing by $h^2$, we  get
\begin{align*}
-\frac{c^2}{3}H_0^2\frac{h_x^3}{h^3}+\frac{1}{3}c^2H_0^2\frac{1}{h^2}h_{xx}h_x
-\frac{1}{3}cv_\infty^\sharp
h_{xx}h_x=&\frac{h-H_0}{2 h^3}(2c^2H_0-gh(h+H_0))h_x\\
&-
(\frac{h^3}{H_0^3}-1)\frac{h_x}{h^2}E_\infty-2\frac{F_\infty}{c}\frac{(h^2-H_0^2)}{H_0^5}hh_x.
\end{align*}
After integrating in $x$, this yields
\begin{align*}
\frac{1}{6}c\big(cH_0^2-v_\infty^\sharp h^2)\frac{h_x^2}{h^2} &=
\frac{c^2}{2}\frac{1}{h^2}(h-H_0)^2-\frac{1}{2}g\frac{1}{h}(h-H_0)^2
\\&-\frac{1}{2H_0^3}\frac{1}{h}(h-H_0)^2(h+2H_0)E_\infty
-\frac{F_\infty}{2c}\frac{(h^2-H_0^2)^2}{H_0^5}
\end{align*}
(the integration constant has been chosen in order to respect the
constraint that $\zeta$ and its derivatives vanish at
infinity), or equivalently
\begin{equation}\label{bdx3}
\frac{c}{3}\big(cH_0^2-v^\sharp_\infty h^2\big)h_x^2=(h-H_0)^2\Big(c^2-gh-
\frac{h(h+2H_0)}{H_0^3}E_\infty-\frac{h^2(h+H_0)^2}{H_0^5}\frac{F_\infty}{c}\Big).
\end{equation}

{\bf Step 2.} Expressions for the velocity and qualitative analysis. By definition, if $h$ is
the profile of a solitary wave then it is a $C^2$-function and its
derivative must vanish at its maximum. The formula \eqref{bdx3} then
provides directly the only two (recall that $F_\infty$ is assumed to
be zero here) possible values for the speed $c$. \\
Since the function $h\mapsto gh+\frac{h(h+2H_0)}{H_0^3}$ is strictly
increasing on $\R^+$, we also deduce from \eqref{bdx3} that $h_x$
cannot vanish at another point, and therefore that the maximum of $h$ is
attained at a unique point $x_{\rm max}$, and further, that $h$ has to
be monotonous on both sides of $x_{\rm max}$. The fact that it is
decaying on $x>x_{\rm max}$ follows from the condition that $h\to
H_0<h_{\rm max}$ at infinity. Finally, the fact that $h$ is symmetric
with respect to $x_{\rm max}$ follows from the simple observation that
if $h$ solves \eqref{ODEbase} for $x\geq x_{\rm max}$ with boundary
conditions $h(x_{\rm max})=h_{\rm max}$ and $h_x(x_{\rm max})=0$ then
$x\mapsto h(2x_{\rm max}-x)$ furnishes a solution for $x\leq x_{\rm max}$.\\

{\bf Step 3.} Existence of a solitary wave of speed $c=\underline{c}$
or $c=-\underline{c}$. If ${c}H_0^2-v_\infty^\sharp h_{\rm max}^2\neq
0$, then the Cauchy-Lipschitz theorem furnishes a local solution
with boundary condition $h(x_{\rm max})=h_{\rm max}$ and $h_x(x_{\rm
  max})=0$. If moreover $c({c}H_0^2-v_\infty^\sharp h_{\rm max}^2)>0$, then it is easy to deduce
from \eqref{ODEbase} that this local solution satisfies $h''(x_{\rm
  max})<0$ and therefore that the solution attains a local maximum at
$x_{\rm max}$. Proceeding as in Step 2, one gets that this local
solution is symmetric with respect to $x_{\rm max}$ and decaying on
$x>x_{\rm max}$. Moreover, one always has $h>H_0$; indeed, if one had
$h(x_0)=H_0$ for some $x_0\in \R$, then one would have $h_x(x_0)=0$ by
\eqref{bdx3}, and by uniqueness, one would have $h\equiv H_0$, which
is absurd. Therefore $h$ decays to some limit as $x\to \infty$, and
this limit is necessarily $H_0$ by \eqref{bdx3}. The identity
\eqref{bdx3} also shows that $h_x$ remains bounded, so that no blow up
of $h$ nor $h_x$ can occur and the solution of the ODE \eqref{ODEbase}
is global.\\

{\bf Step 4.} Non existence  of a solitary wave of speed $c=\underline{c}$
or $c=-\underline{c}$. If  $c({c}H_0^2-v_\infty^\sharp h_{\rm max}^2)<0$, then it is easy to deduce
from \eqref{bdx3} that no solitary wave can exist.  The only case left to investigate is
therefore the critical case ${c}H_0^2-v_\infty^\sharp h_{\rm
  max}^2=0$. In this case, one gets from \eqref{ODEbase} that
\begin{align*}
\frac{c^2}{3}H_0^2 \frac{h'(x_{\rm max})}{h(x_{\rm max})}&=\frac{h_{\rm
  max}-H_0}{2h_{\rm max}}(gh_{\rm max}(h_{\rm max}+H_0)-2c^2
H_0)+(\frac{h_{\rm max}^3}{H_0^3}-1)E_\infty\\
&>0
\end{align*}
which contradicts the assumption that $h$ is a $C^1$-function
attaining its maximum at $x_{\rm max}$.
\end{proof}

\vspace{0.2cm}
\begin{rem}\label{remSolWav1}
If in addition to the assumption $F_\infty=0$ we take
$E_\infty=v_\infty^\sharp=0$ in the statement of Proposition
\ref{propSolWav}, then one has
  $\underline{v}^\sharp=\underline{E}=\underline{F}=0$ in Definition
  \ref{defSolWav} so that the solitary waves are the same as in the
  irrotational setting for which it is well known that explicit
  solitary waves exist. More precisely, for {\it any} maximal
  amplitude $h_{\rm max}>H_0$, there exists two (up to translations)
  solitary waves of
  speed $c=\pm \sqrt{gh_{\rm max}}$ and with the same profile $h_+=h_-$  given by the resolution of \eqref{ODEbase}
and which can in this particular case be computed explicitly,
\beq\label{solitary:gn}
h(x)=H_0+\eps H_0 (\mbox{\rm sech}(\frac{x}{\lambda}))^2
\quad\mbox{ with }\quad
\lambda= \frac{2}{\sqrt{3}} \sqrt{\frac{1+\eps}{\eps}}H_0 ,
\eeq
and where we denoted $h_{\rm max}=H_0(1+\eps)$.
\end{rem}

\vspace{0.2cm}
\begin{rem}\label{remSolWav2}
If in addition to the assumption $F_\infty=0$ we take
$v_\infty^\sharp=0$  but consider the case $E_\infty>0$, the situation is
  qualitatively the same as in Remark \ref{remSolWav1}: for {\it any} maximal
  amplitude $h_{\rm max}=H_0(1+\eps)>H_0$, there exists two solitary waves of
  same shape and of  opposite speed $c=\pm \uc$. The only difference is
  that the speed $\uc$ is larger than in the irrotational case, 
\beq\label{eq:speed_sol1}
\uc=\Big(gh_{\rm max}+\frac{ h_{\rm max} (h_{\rm max}+2H_0)}{H_0^3}E_\infty\Big)^{1/2},
\eeq
and that the solitary waves becomes narrower as $E_\infty$
increases (this follows easily from the comparison principle for
ODEs); see Figure \ref{influence:E}. 

\begin{figure}
\centering
\includegraphics[width=1\textwidth]{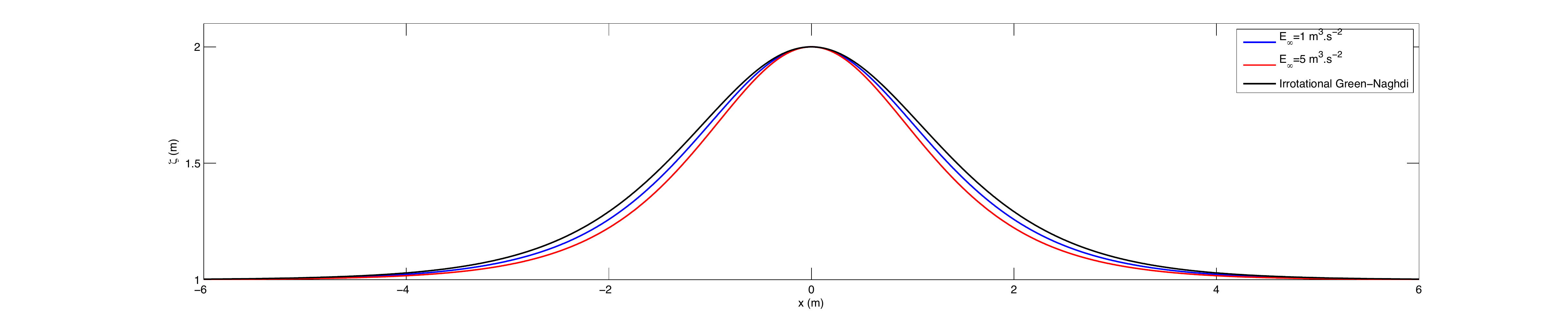}
\caption{Influence of $E_\infty$ on the solitary wave profile for $H_0=1\,m$, $h_{\rm max}=2\,m$, $v^\sharp_\infty=0$, $F_\infty=0$.}
\label{influence:E}
\end{figure}
\end{rem}

\begin{rem}\label{remSolWav3}
 If we assume that $F_\infty=0$ but $E_\infty>0$ and  $v_\infty^\sharp
 > 0$ (the case $v_\infty^\sharp
 <0$ can be treated in a similar way) in the statement of Proposition \ref{propSolWav} then there are
 two major qualitative changes with respect to the situation
 considered in Remark \ref{remSolWav2}. The
  first one is that 
  right going solitary waves do not exist for any maximal amplitude
  $h_{\rm max}>H_0$. Indeed, the criterion given in the second point
  of the proposition is always satisfied for the left-going solitary
  wave, but requires for the right-going one that
$$
\underline{c}H_0^2-v_\infty^\sharp h_{\rm max}^2>0.
$$
or equivalently, using the explicit expression of $\underline{c}$
given in Proposition \ref{propSolWav},
$$
\big( gH_0+(\tilde h+2)\frac{E_\infty}{H_0}\big)>(v_\infty^\sharp)^2\tilde
h^3
\quad\mbox{ with }\quad \tilde{h}=\frac{h_{\rm max}}{H_0}.
$$
In the case where $v_\infty^\sharp>0$, the criterion given in the statement of Proposition \ref{propSolWav}
for the existence of solitary waves can therefore be
restated as: left-going solitary waves always exist, but right-going
solitary waves exist if and only if $h_{\max}<h_{\rm crit}$ where the
critical height $h_{\rm crit}$ is given by $h_{\rm crit}=H_0 \tilde h_{\rm crit}$
 with 
$\tilde{h}_{\rm crit}$ the only positive root of the polynomial 
$$
P(X)=(v_\infty^\sharp)^2 X^3-\big( gH_0+(X+2)\frac{E_\infty}{H_0}\big).
$$

\begin{figure}
\centering
\includegraphics[width=1\textwidth]{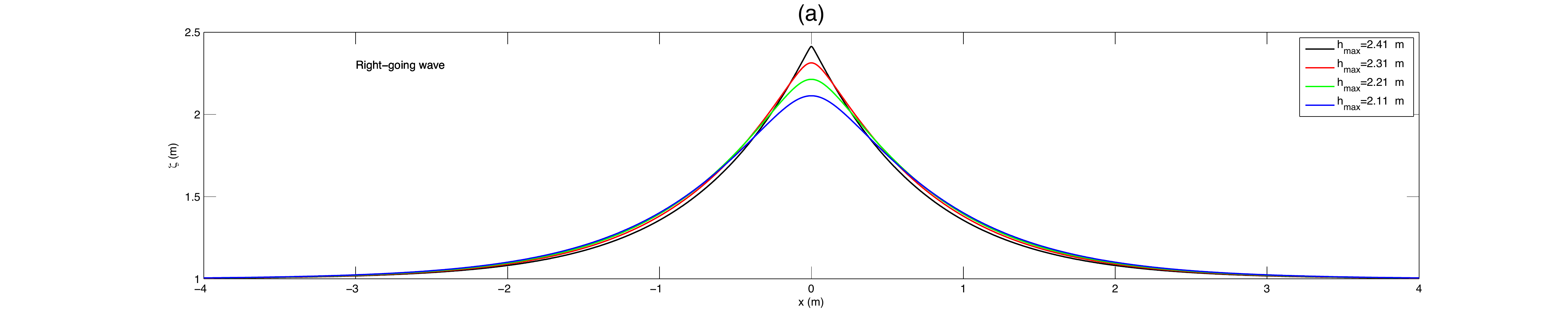}
\includegraphics[width=1\textwidth]{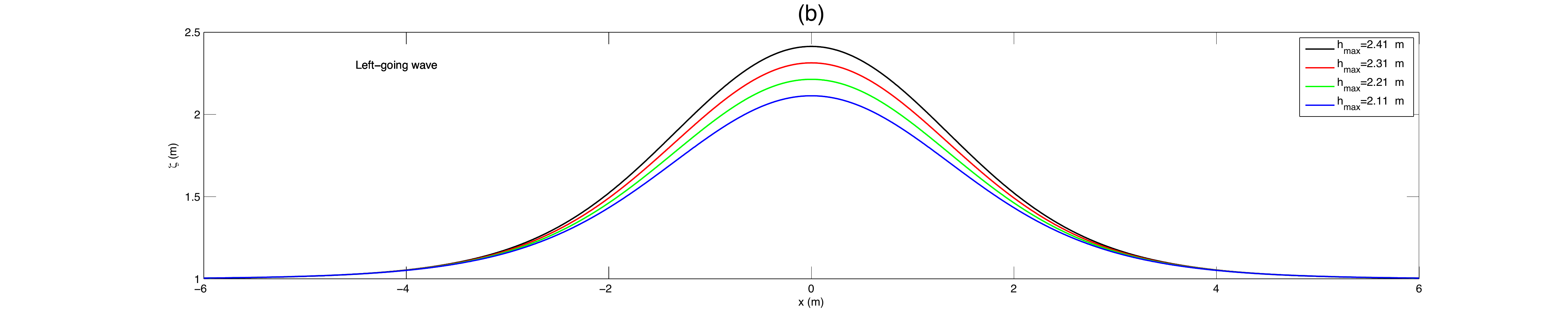}
\includegraphics[width=1\textwidth]{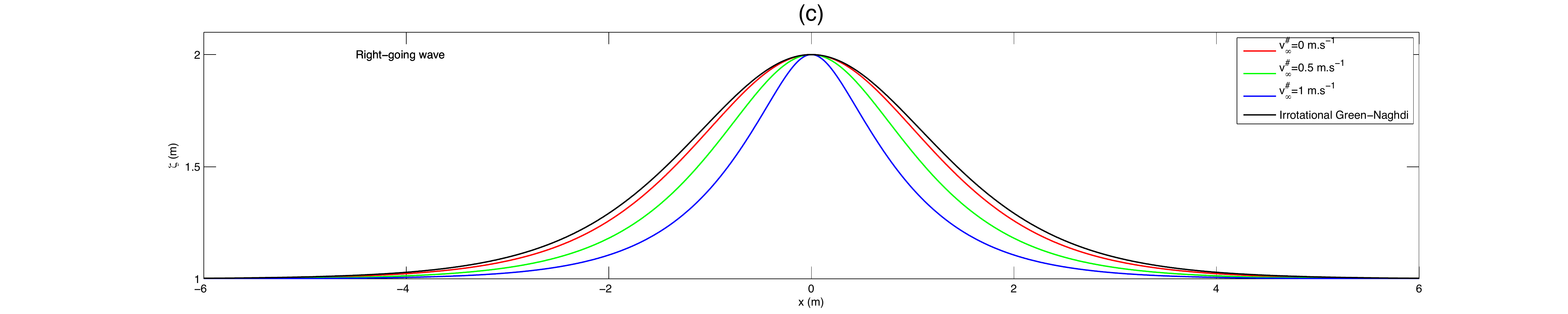}
\includegraphics[width=1\textwidth]{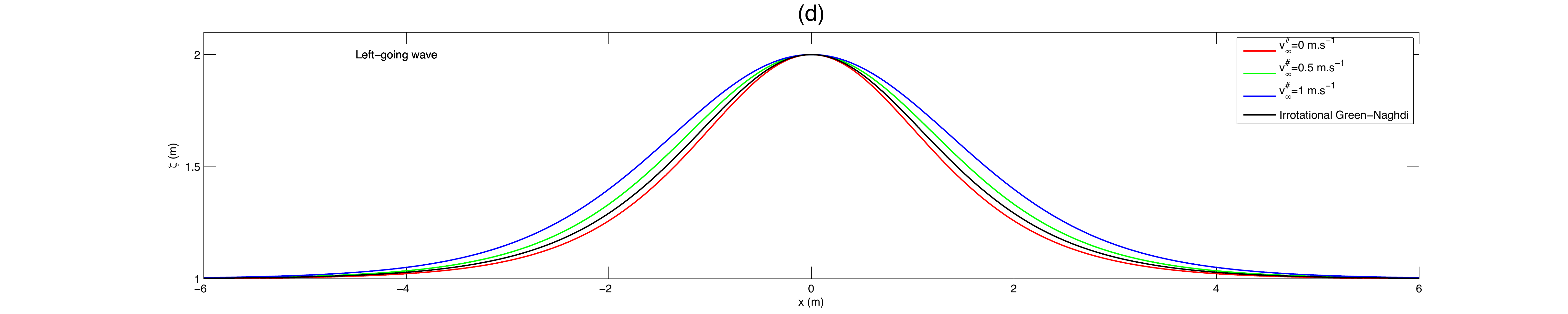}
\caption{Influence of $v^\sharp_\infty>0$ on the profiles for $H_0=1\,m$, $E_\infty=1\,m^3.s^{-2}$ and $F_\infty=0$. (a) Right-going solitary waves of nearly critical amplitudes for $v^\sharp_\infty=1\,m.s^{-1}$ and therefore $h_{\rm crit}\approx 2.42\,m$. (b) Left going solitary waves for the same amplitudes. (c) Right-going solitary waves for $h_{\rm max}=2\,m$ and increasing values of $v^\sharp_\infty$. (d) Left-going solitary waves for $h_{\rm max}=2\,m$ and increasing values of $v^\sharp_\infty$.}
\label{figcrit}
\end{figure}
The second qualitative change with respect to the situation previously
considered is that the shape of the two solitary waves of speed $c=\pm
\uc$ are not the same since the ODE in Definition \ref{defSolWav} does no
longer depend on $c$ through $c^2$ only. We refer to
Figure \ref{figcrit} for an illustration of this behavior, in which we show the right and left-going waves profiles for increasing values of $h_{\rm max}$ in the vicinity of $h_{\rm crit}$, for  $v_\infty^\sharp =1\,m.s^{-1}$. We also highlight the influence of increasing values of $v_\infty^\sharp$ on the left-going waves profiles for a given value of $h_{\rm max}$.
\end{rem}

\vspace{0.2cm}
We recall that we assumed in Proposition \ref{propSolWav} that
$F_\infty=0$. Let us now give a brief discussion about the general
case $E_\infty>0$,
$v_\infty^\sharp > 0$ and $F_\infty>0$ (one can treat the case
$v_\infty^\sharp<0$ and/or $F_\infty<0$ in a similar way).  The
  presence of $F_\infty\neq 0$ implies that the possible speeds for
  solitary waves of maximal amplitude $h_{\rm max}$ are found by solving the {\it third} order polynomial
\begin{align}\label{eqc}
X^3+pX+q \quad \mbox{ with }\quad & p=-\big(g h_{\rm max}+h_{\rm max} (h_{\rm max}+2H_0)\frac{E_\infty}{H_0^3}\big)\\
 & q=-h_{\rm max}^2 (h_{\rm max}+H_0)^2 \frac{F_\infty}{H_0^5},
\end{align}
(this is a simple consequence of \eqref{bdx3}). Defining as previously
$\eps$ by
$$
\frac{h_{\rm max}}{H_0}=1+\eps
$$
the discriminent $\Delta=-(4p^3+27q^2)$ of this polynomial is always
positive provided that the following smallness condition holds for
$F_\infty$
\begin{equation}\label{condF}
F_\infty^2<\frac{27}{4}\frac{H_0^5}{(1+\eps)(2+\eps)^4}\big( g+(3+\eps)\frac{E_\infty}{H_0^2} \big)^3
\end{equation}
(this condition is satisfied for all realistic configurations). The
polynomial \eqref{eqc} has then three distinct
roots. Since the coefficient of $X^2$ is equal to zero, the sum of the
three roots is necessarily equal to zero; moreover, their product has the sign of $-q$, and therefore the
sign of $F_\infty$. If $F_\infty>0$, then one has
one positive root $0<c_{+}$  and two negative roots
$-c_{-,2}<-c_{-,1}<0$.  There are therefore  possibly two left going
  solitary waves, and a right going one. The right-going wave is
  subject to the same constraint $h_{\rm max}< h_{\rm crit}$ as in the case
  $E_\infty>0$, $v_\infty^\sharp > 0$ and $F_\infty=0$. In addition,
  \eqref{bdx3} shows that the function
\begin{equation}\label{defphi}
\varphi_c: h\mapsto c^2-h-\eps^2\mu
h(h+2)E_\infty-\eps^3\mu^{3/2}h^2(h+1)^2\frac{F_\infty}{c}
\end{equation}
must be positive for all $H_0\leq h <h_{\rm max}$. We can now state
the following proposition where for the sake of simplicity, we
considered only the case $v_\infty^\sharp>0$ and $F_\infty>0$. The
cases where these quantities are negative can be treated similarly.

\vspace{0.2cm}
\begin{proposition}\label{proposition:2}
Let $E_\infty>0$, $v^\sharp_\infty>0$ and $F_\infty>0$. Let also
$h_{\rm max}=H_0(1+\eps)$ with $\eps>0$ and assume that \eqref{condF}
is satisfied. \\
{\bf i.} Up to translations, there can be at most three solitary waves of maximal height
$h_{\rm max}$ for \eqref{GNvort1dfin} and corresponding speeds
$-c_{-,2}<-c_{-,1}<0<c_+$ given by the roots of \eqref{eqc}.\\
{\bf ii.} The solitary wave of speed $c_+$ exists if and only if the following conditions hold
\begin{equation}\label{comp:p2}
{c_+}H_0^2-v_\infty^\sharp h_{\rm max}^2>0 \quad\mbox{ and }\quad
\forall h\in [H_0,h_{\rm max}),\qquad \varphi_{c_+}(h)>0,
\end{equation}
where $\varphi_{c}$ is as defined in \eqref{defphi}.\\
{\bf iii.} The solitary wave of speed $c_{-,j}$ ($j=1,2$) exists if
and only $\varphi_{-c_{-,j}}>0$ on $[H_0,h_{\rm max})$.
\end{proposition}
\begin{rem}
Though there could be in principle a third solitary wave arising in
the case $F_\infty>0$, we could not exhibit any configuration where
this is the case because the condition $\varphi_{-c_{-,1}}>0$ on
$[H_0,h_{\rm max})$ is never fulfilled. In practice, there are as in
the case $F_\infty=0$ one left going
and one right going solitary wave, of different shape and of
respective speed $-c_{-,2}$ and $c_+$. The profiles of the corresponding solitary waves are shown on Figure \ref{influence:F}. Note that smaller values of $F_\infty$ have to be taken to obtain the profiles of the left-going waves (Figure \ref{influence:F} (b)) in order to fulfill the condition $\varphi_{-c_{-,2}}>0$ on $[H_0,h_{\rm max})$. Note also that additional solitary waves profiles can be observed, for both the constant vorticity model \eqref{gn:const} and the general model \eqref{gn:vort:dim}, in \S\ref{validation}.
\end{rem}

\begin{figure}
\centering
\includegraphics[width=1\textwidth]{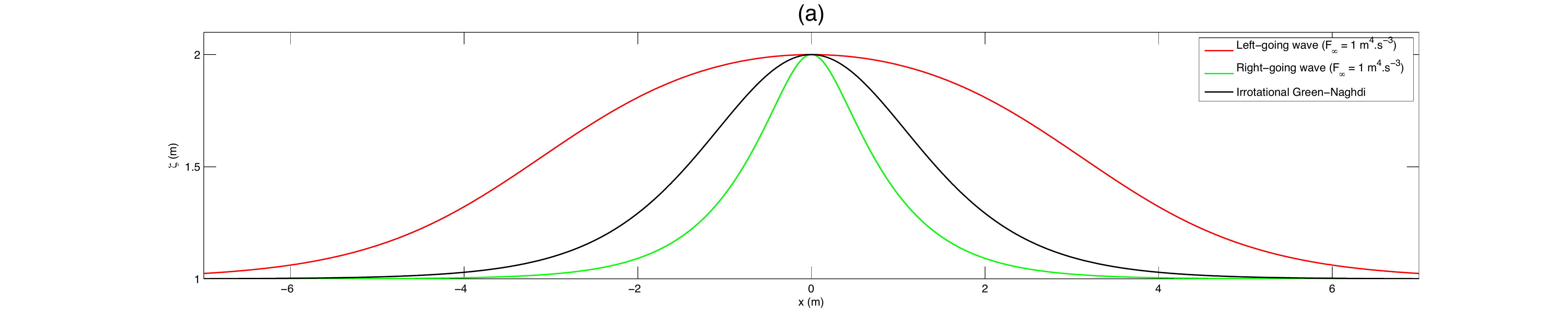}
\includegraphics[width=1\textwidth]{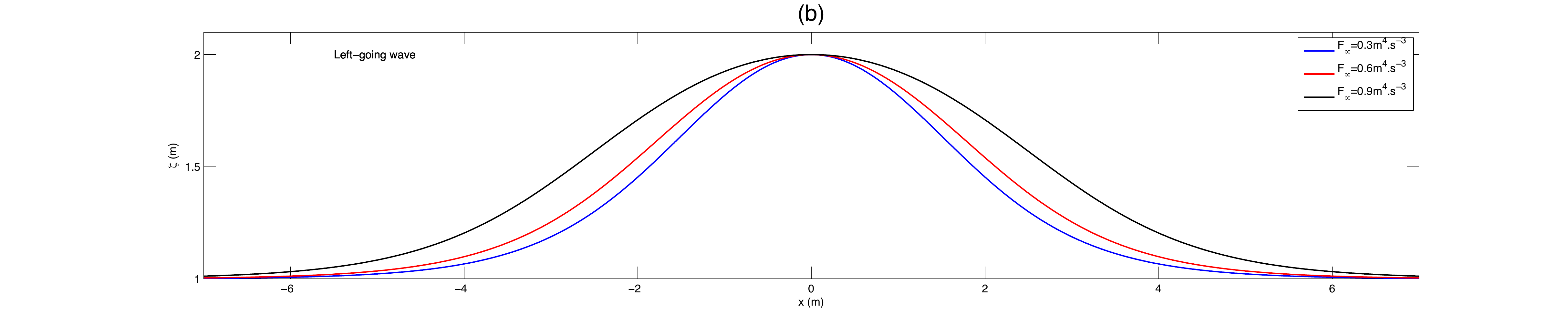}
\includegraphics[width=1\textwidth]{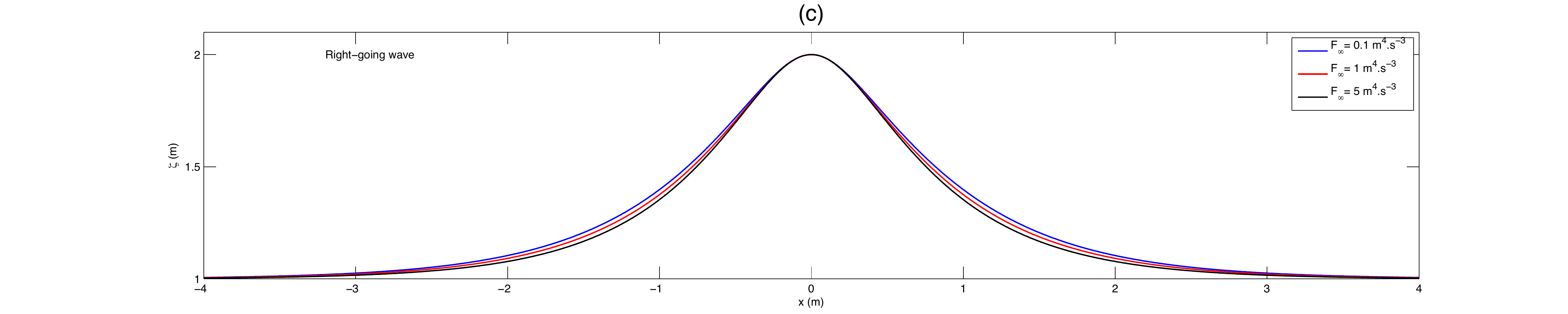}
\caption{Influence of $F_\infty$ on the solitary waves profiles. (a) Shape of the left and right going solitary waves, and comparison with the usual Green-Naghdi solitary
  wave, for $H_0=1\,m$, $h_{\rm max}=2\,m$, $E_\infty=1\,m^3.s^{-2}$, $v_\infty=1\,m.s^{-1}$, $F_\infty=1\,m^4.s^{-3}$. (b) Influence of the value of $F_\infty$ on the left-going wave's profile. (c) Influence of the value of $F_\infty$ on the right-going wave's profile.}
\label{influence:F}
\end{figure}

\subsection{Existence of peaked solitary waves}\label{ode:analysis:peak}

As seen in the previous section, when $E_\infty>0$,
$v_\infty^\sharp>0$ and $F_\infty=0$, solitary waves have speed $\pm
\underline{c}$, with
$$
\underline{c}=\underline{c}(h_{\rm max})=\Big(gh_{\rm max}+\frac{h_{\rm max}(h_{\rm max}+2H_0)}{H_0^3}E_\infty\Big)^{1/2}
$$
and the maximal amplitude $h_{\rm max}$ of the right-going solitary
wave cannot exceed a critical value $h_{\rm crit}$ corresponding to the
only positive root of the equation
$$
\underline{c}(h_{\rm crit})H_0^2-v_\infty^\sharp h_{\rm crit}^2=0.
$$
Figure \ref{figcrit}
suggests that the shape of the solitary waves tend to form an angle at
their crest as their amplitude become close to the maximal
amplitude. A byproduct of the analysis of the previous section is that there cannot exist any
{\it smooth} solitary wave of maximal amplitude $h_{\rm
  crit}$. However, we show here that it is possible to obtain a {\it
  peaked} solitary wave of maximal amplitude in the following sense:
  
  \vspace{0.2cm}
\begin{definition}\label{defSolWavpeak}
A \emph{peaked} solitary wave of speed $c$, centered at $x_0\in
\R$, for \eqref{GNvort1dfin}, is a mapping
$$
(t,x)\in \R^2\mapsto (\underline{\zeta},\underline{\ovv}, \underline{v}^\sharp,
\underline{E}, \underline{F})(x-ct-x_0)
$$
 such that there exists $h\in C(\R)$, with $h_{\vert_{\R^+}}\in
 C^2([0,\infty))$, $h_{\vert_{\R^-}}\in
 C^2((-\infty,0]) $, and
$E_\infty>0$, $v^\sharp_\infty\in \R$ and $F_\infty\in \R$ such that
\begin{align}
&\underline{\zeta}=h-H_0,\quad \underline{\ovv}=c\frac{h-H_0}{h},\quad
\underline{v}^\sharp=\frac{h}{H_0}v^\sharp_\infty,\\
&\underline{E}=\frac{h^3}{H_0^3}E_\infty+2\frac{F_\infty}{c}\frac{(h^2-H_0^2)h^3}{H_0^5},\quad
\underline{F}=\frac{h^4}{H_0^4} F_\infty,
\end{align}
and such that $h$ solves the ODE \eqref{ODEbase}
on $\R^+$ and $\R^-$, and satisfies $\lim_{\pm \infty} h=H_0$. 
\end{definition}

\vspace{0.2cm}
The proposition below proves the existence of peaked solitary waves in
the case $F_\infty=0$. Such a property could also be established for
$F_\infty\neq 0$ (according to Proposition \ref{proposition:2} there is also a critical
maximal wave in some cases when $F_\infty\neq 0$), but the proof would
be more technical and since no new phenomena arises in this case, we
decide not to treat it. Some examples of peaked solitary waves are plotted on Figure \ref{fig:peakon}, on which we highlight the influence of the value of $v_\infty^\sharp>0$ on the critical amplitude $h_{\rm crit}$. 

\vspace{0.2cm}
\begin{proposition}\label{prop:peaked}
Let $E_\infty>0$, $v_\infty^\sharp>0$ and $F_\infty=0$. For all $x_0\in \R$, there exists a unique peaked solitary wave
centered at $x_0$ of
critical maximal amplitude $h_{\rm crit}$ and speed
$\underline{c}=\underline{c}(h_{\rm crit})$. It is even, decaying on
both sides of the crest, and its angle
at the crest is $2\theta$, with
$$
\tan \theta=\Big(\frac{3}{2}\frac{E_\infty}{\uc H_0^3v_\infty^\sharp}
(h_{\rm crit}-H_0)^2\big(1+\frac{(v_\infty^\sharp)^2}{E_\infty}\frac{h_{\rm
      crit}^2}{H_0}\big)\Big)^{-1/2}.
$$
\end{proposition}
\begin{proof}
We focus here on the case $x\geq 0$; the case of negative values of
$x$ can be treated similarly. Without loss of generality, we also
assume that $x_0=0$. By definition of $h_{\rm crit}$, one can write, for all $h$,
$$
\underline{c}H_0^2-v_\infty^\sharp
h^2=-{v_\infty^\sharp}(h-h_{\rm
  crit})(h+h_{\rm crit})
$$
and
$$
\underline{c}^2-gh-\frac{h(h+2H_0)}{H_0^3}E_\infty=-\frac{E_\infty}{H_0^3}(h-h_{\rm
  crit})(h+\frac{(v_\infty^\sharp)^2}{E_\infty}\frac{h_{\rm crit}^3}{H_0}).
$$
On can therefore rewrite \eqref{bdx3} under the form
$$
\frac{\uc}{3}v_\infty^\sharp(h-h_{\rm crit})(h+h_{\rm crit})h_x^2=\frac{E_\infty}{H_0^3} (h-H_0)^2 (h-h_{\rm
  crit})(h+\frac{(v_\infty^\sharp)^2}{E_\infty}\frac{h_{\rm crit}^3}{H_0})
$$
or equivalently
$$
\frac{\uc}{3}v_\infty^\sharp(h+h_{\rm crit})h_x^2=\frac{E_\infty}{H_0^3} (h-H_0)^2(h+\frac{(v_\infty^\sharp)^2}{E_\infty}\frac{h_{\rm crit}^3}{H_0}).
$$
Since $h_{\rm crit}$ is by definition the maximal value of $h$ and
since $h$ cannot reach the value $H_0$ (otherwise it would be
identically equal to $H_0$), this ODE is equivalent to
$$
h_x=-\Big(3\frac{E_\infty}{\uc H_0^3v_\infty^\sharp}
(h-H_0)^2\frac{h+\frac{(v_\infty^\sharp)^2}{E_\infty}\frac{h_{\rm
      crit}^3}{H_0}}{h+h_{\rm crit}}\Big)^{1/2}.
$$
Existence of a local solution is therefore given by the standard
Cauchy-Lipschitz theorem; the fact that the solution is global and
tends to $H_0$ at infinity is then easily established as in Step 3 of
the proof of Proposition \ref{propSolWav}. When evaluated at the origin (i.e
replacing $h$ by $h_{\rm crit}$) in the above formula
\end{proof}

\begin{figure}
\centering
\includegraphics[width=1\textwidth]{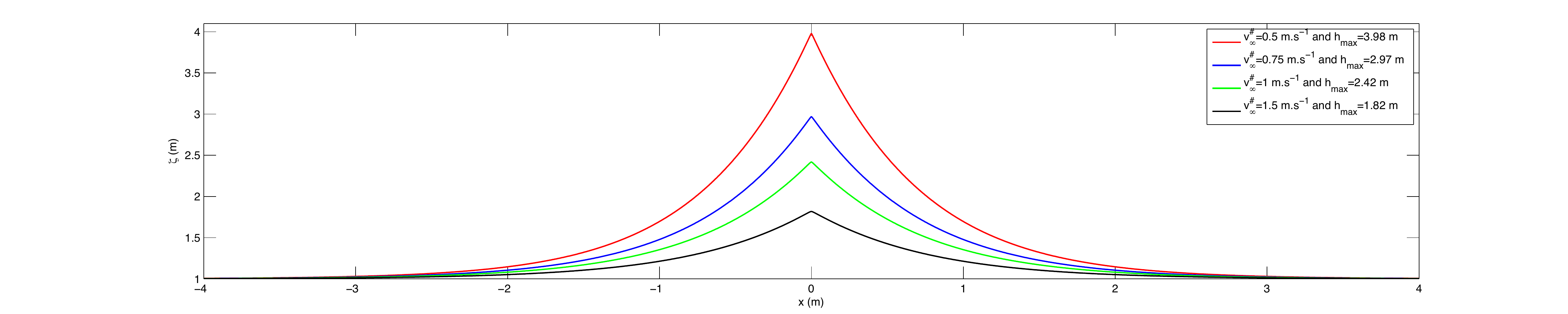}
\caption{Influence of $v^\sharp_\infty>0$ on the peaked solitary waves profiles for $H_0=1\,m$, $E_\infty=1\,m^3.s^{-2}$ and $F_\infty=0$. }
\label{fig:peakon}
\end{figure}

\section{Numerical method}\label{sect:num:scheme}

We introduce now a simple numerical method to approximate the solutions of system \eqref{gn:vort:dim} and illustrate the propagation of some of the various wave profiles exhibited in \S \ref{section:sol_wave}. This approach is inspired by some of our previous works \cite{Bonneton20111479, MR2811693, lannes_marche:2014}, and we give some details about the new ingredients introduced to account for the specificities of \eqref{gn:vort:dim}. 

\subsection{Splitting}\label{sect:splitting}
To build our numerical method, denoting $E=h^2\tilde{E}$ et $F=h^3\tilde{F}$ and using the mass conservation equation, we rewrite system (\ref{gn:vort:dim}) in the equivalent form:

\be
\begin{cases}
&h_t + (h\bar{v})_x = 0,\\
&(1+\mathbb{T})\Bigr((h\bv)_t + (h\bv^2)_x\Bigr) + gh\zeta_x+ (h^2\tE)_x + h\cQ_1(\bv) + h\cC(\bv,\dv) =0,\\
&v_t^\sharp + (\bv\dv)_x =0,\label{gn:vorti:dim:2}\\
&\tE_t + (\bv\tE)_x+ 3\tF h_x + h\tF_x  = 0,\\
&\tF_t + (\bv\tF)_x = 0.
\end{cases}
\ee
Then, in the spirit of \cite{Bonneton20111479, MR2811693, lannes_marche:2014}, we decompose the solution operator $S(\cdot)$ associated to the formulation \eqref{gn:vorti:dim:2}, at each time step by the second order splitting scheme
\begin{equation}\label{S0}
	S(\delta_t)=S_1(\delta_t/2)S_2(\delta_t)S_1(\delta_t/2),
\end{equation}
where $S_1$ and $S_2$ are respectively associated to the transport part and dispersive perturbation of
the Green-Naghdi equations (\ref{gn:vorti:dim:2}). More precisely:
\begin{itemize}
\item  $S_1(t)$ is the solution operator
associated to the conservative \textit{propagation} step
\be
	\label{S1:gen}
	\left\lbrace
	\begin{array}{lcl}
	\vspace{0.2cm}
	\dsp h_t+(h\bv)_x&=&0,\\
	
	\vspace{0.2cm}

	\dsp (h\bv)_t +  (h\bv^2)_x + gh\zeta_x + (h^2\tE)_x &=&0,\\
	\vspace{0.2cm}
	v_t^\sharp  +  (\bv\dv)_x &=&0,\\
	
	\vspace{0.2cm}
	\tE_t + (\bv \tE)_x &=& 0,\\
         \tF_t + (\bv \tF)_x &=& 0.
	\end{array}\right.
\ee
\item $S_2(t)$ is the solution operator associated to the \textit{dispersive correction},
\be
	\label{S2:gen}
	\left\lbrace
	\begin{array}{lcl}
	\vspace{0.2cm}
	\dsp h_t =0,\\
	
	\vspace{0.2cm}
	\dsp (h \bv)_t - gh \zeta_x -  (h^2\tE)_x + (1+ \mathbb{T})^{-1}\big[gh\zeta_x + (h^2\tE)_x +  h\cQ_1(\bv) + h\cC(\bv,\dv)\big]=0,\\
	
	\vspace{0.2cm}
	v_t^\sharp  =0,\\
	
	\vspace{0.2cm}
        \tE_t   + 3\tF h_x + h\tF_x= 0,\\
        \tF_t  = 0.
	\end{array}\right.
\ee

\end{itemize}

As detailed in the next subsection, $S_1(t)$ is discretized using a finite-volume approach, and we use a finite-difference approach for $S_2(t)$. 

\vspace{0.2cm}
\begin{rem}
The proposed splitting scheme is classically of second order accuracy in time, and it is shown in \cite{Bonneton20111479} that the corresponding semi-discretized dispersion relation approaches the exact dispersion relation of the Green-Naghdi  at order $2$ in time. Of course, other approaches can be used, like the unsplitted approaches proposed in \cite{DuranMarche:2014ab} or \cite{mario:2015,popinet} for instance.
\end{rem}
\vspace{0.2cm}

\medbreak

\noindent
We use the following notations in the following:
\begin{itemize}
\item the numerical one-dimensional domain $\Omega$ is uniformly divided into $N_x$ cells $(\cC_i)_{1 \le i \le N_x}$ such that $\cC_i = [x_{i-\demi}, x_{i+\demi}]$, where  $(x_{i+\demi})_{0 \le i \le N_x}$ are the $N_x+1$ nodes of the regular grid. We denote by $x_i$ the center of $\cC_i$,
\item  we denote by $\delta_x$ the cell size (constant in this work) and by $\delta_t$ the chosen time step (to be specified according to a relevant CFL-like condition),
\item we denote by $\bar{w_i}^n$ the averaged value of an arbitrary quantity $w$ on the $i^{th}$ cell $\cC_i$ at time $t_n = n\delta_t$.
\end{itemize}

\medbreak

\subsection{Discretization of the conservative step}\label{sect:conservative}

We focus on the discretization of system \eqref{S1:gen} which can be written in compact form as follows
\beq\label{syst:compact}
\dt \cW + \dx \eF(\cW)  = \eS(\cW,b),
\eeq
with $\cW = (h, h\bv, \dv, \tE, \tF)$ and 
\beq
\eF(\cW) = 
(
 h\bv,\,
 h\bv^2 + p(h,\tE),\,
\bv\dv,\,
 \bv \tE,\,
 \bv \tF
 ),\;\;\;
\eS(\cW,b) = 
(
 0,\,
-gh b_x,\,
0,\,
 0,\,
0
),
    \eeq
    
 \noindent
with $p(h,\tE)=\dsp\frac{g}{2}h^2 + h^2\tE$. Neglecting the bottom variations, the study of the associated algebra shows that the system is hyperbolic with the following eigenvalues:
\begin{align}
\lambda_1 = \bv -\sqrt{gh + 3h\tE},\;\;\; \lambda_2 = \bv +\sqrt{gh + 3h\tE}, \;\;\;\lambda_3=\lambda_4=\lambda_5=\bv.
\end{align}
\quad

\subsubsection{$1^{st}$-order FV discretization of the homogeneous system}\label{FO}

We first study a first order conservative spatial discretization of the homogeneous system associated with (\ref{syst:compact}):
\be\label{vf:scheme}
\bar{\cW}_i^{n+1}-\bar{\cW}_i^{n} +   \frac{\delta t}{\delta_x}   \Bigl( \cF  \bigl(\bar{\cW}_{i}^n,  \bar{\cW}_{i+1}^n\bigr)   -  \cF \left(\bar{\cW}_{i-1}^n,\bar{\cW}_{i}^n \right) \Bigr)= 0
\ee
where $(u,v)\mapsto\cF(u,v)$ is a numerical flux function consistent with the physical flux $w\mapsto\eF(w)$. For the numerical validations shown in \S\ref{validation}, and considering the $3$ waves structure of the algebra, we have implemented a HLLC-type Riemann solver, see for instance \cite{BattenClarke:1997aa, BOUCHUT:2004}. 

\subsubsection{Discretization of the topography and high-order extension}

The discretization of the topography source term occurring in \eqref{syst:compact}  is done following the well-balanced approach for the Saint-Venant equations described in \cite{Liang2009873}, allowing to preserve the motionless steady states corresponding to 
\beq\label{steady:states}
\zeta = 0,\quad \bv=0,\quad \dv = 0,\quad E=0, \quad F=0.
\eeq
Note that one of the main properties of this approach is that whenever the initial solver satisfies some classical stability properties, it yields a simple and fast well-balanced scheme that preserves the positivity of the water height.\\
As shown in previous studies \cite{MR2811693, Bonneton20111479, DuranMarche:2014ab, lannes_marche:2014}, the use of high-order schemes is mandatory for the study of dispersive water waves, to avoid as much as possible to pollute the dispersive properties of the model with some dispersive truncation errors associated with $2^{nd}$ order schemes. Based on discrete finite-volume cell averaging $\bar{\cW}_i^n$ at time $t^n=n\delta t$ we use in this work $3^{rd}$ and $5^{th}$-order accuracy WENO reconstructions, following \cite{jiang:1996p1236}, together with the weight splitting method \cite{584782}.

\subsection{Spatial discretization of $S_2(\cdot)$, time discretization and boundary conditions}\label{sect:S2}
\label{sectdisp}

Following the approach developed in \cite{Bonneton20111479}, system (\ref{S2:gen}) is discretized using $4^{th}$ finite-differences. The resulting matrix, for the discretization of the linear operator $\mathbb{T}$ is the same as in \cite{Bonneton20111479}. As far as time discretization is concerned, we choose to use explicit methods. The systems corresponding to $S_1$ and $S_2$ are integrated in time using third or fourth-order SSP-\textit{Runge-Kutta} scheme \cite{gottliedtadmor}. For the sake of simplicity, we only use periodic and Neumann boundary conditions here, adapting the ghost-cells methods detailed in \cite{Bonneton20111479}.
\vspace{0.3cm}
\begin{rem}
The whole numerical strategy can be straightforwardly applied to the \textit{constant vorticity} model \eqref{gn:const}, the \textit{medium amplitude} equations \eqref{gn:vort:dimsimpl} and the reduced model with $F=0$ \eqref{modeleFzero}, adapting the approximate Riemann solvers to the corresponding hyperbolic part. 

\end{rem}

\section{Numerical validation}\label{validation}

In this section, we use the analysis of solitary waves performed in Section \ref{section:sol_wave} to validate our numerical scheme. Various kinds of smooth solitary waves exhibited in Section \ref{section:sol_wave} are numerically observed in \S \ref{sect:SWP} and used to evaluate the convergence rate. As shown in \S \ref{sect:PSW}, our code is accurate enough to capture also the extremal peaked solitary waves exhibited in \S \ref{ode:analysis:peak}. We then treat in \S \ref{sect:vortshoal} an example with a non trivial topography which allows us to show that vorticity may have a considerable influence on the shoaling phase. The specified convergence rates are obtained using a discrete $L^2$-error.

\subsection{Solitary waves propagation}\label{sect:SWP}

In the next test cases of this subsection, we consider $H_0=1$ and we compute the propagation of several solitary waves in a computational domain of $200\,m$ long. We consider several set of values for the wave amplitude $\eps H_0$ and the triplet $(E_\infty, \dv_\infty, F_\infty$), allowing to cover the various configurations detailed in \S\ref{section:sol_wave}. Unless stated otherwise, we use WENO3 reconstructions, a SSP-RK3 scheme and we set the CFL number to $0.8$.

\subsubsection{Constant vorticity model}

We consider here the Green-Naghdi model \eqref{gn:const} with constant vorticity. We consequently have
$$
\curl \bU=(0,\omega,0)^T\quad \mbox{ with }\quad \omega(t,x,z)=\omega_0=\mbox{cst}.
$$ 

\begin{figure}
\centering
\includegraphics[width=0.45\textwidth]{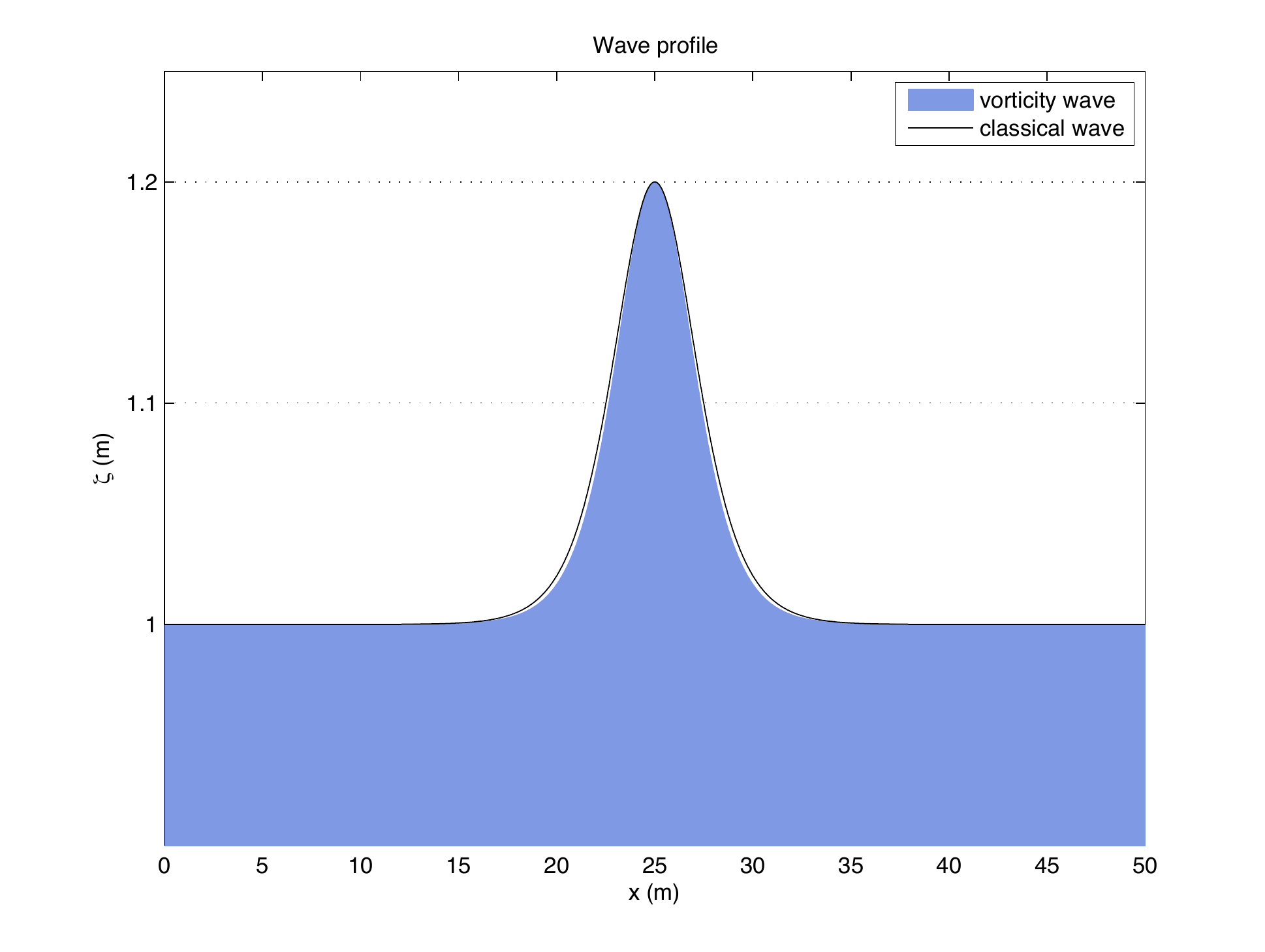}
\includegraphics[width=0.45\textwidth]{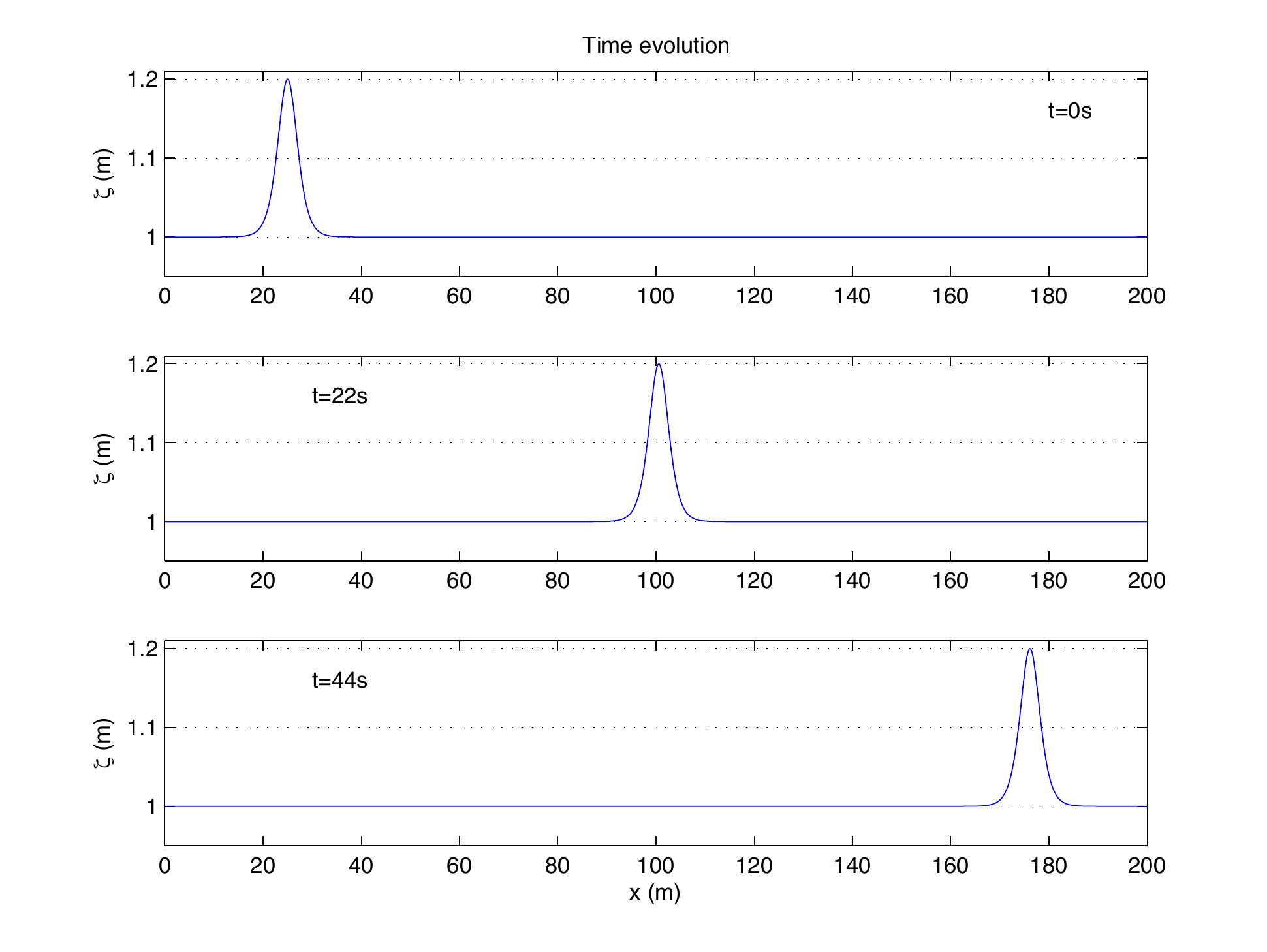}
\caption{Propagation of a right-going solitary wave solution for the constant vorticity model - $\eps=0.2$ and $\omega_0=0.3\,s^{-1}$: initial profile (on the left) and snapshots at different locations along the channel (on the right). The solitary wave solution of the classical Green-Naghdi equations is plotted in black solid line at $t=0\,s$.}\label{test1}
\end{figure}

 The analysis of solitary waves solutions is provided in Proposition \ref{propSolWav}. Indeed, the Green-Naghdi model with constant vorticity can be obtained from the general equations \eqref{gn:vort:dim}, setting $E=\frac{1}{12} h^3 \omega_0^2$, $\dv=h\omega_0$, $F=0$, and consequently neglecting the equations on $\tE$ and $\dv$ which are equivalent to the continuity equation. Solitary wave solutions for \eqref{gn:const} belong to the situation depicted in Remark \ref{remSolWav3}. In the present test case, we set $\omega_0=0.3\,s^{-1}$ and we study the propagation of a solitary wave of relative amplitude $\eps=0.2$, initially centred at $x_0=25\,m$. The initial water height $h^0(x)=h(0,x)$ is computed as a solution of equation \eqref{bdx3} and the corresponding velocity is initialized as
$$
\bv^0(x) = \uc(1-\frac{H_0}{h^0(x)})\quad\mbox{with}\quad
$$
where $\uc$ is obtained from \eqref{eq:speed_sol1}:
\beq\label{eq:speed_sol_const}
\uc^2=gh_{\rm max}+\frac{ h_{\rm max} (h_{\rm max}+2H_0)}{12}\omega_0^2.
\eeq
\begin{rem}
From Proposition \ref{propSolWav}, we observe that with the choice $\omega_0>0$, we have existence of right-going solitary waves (propagating at speed $\uc>0$) for all amplitudes $a$. On the contrary, the existence of left-going waves is ruled by the additional compatibility condition 
$$
\underline{c}H_0- \omega_0 h_{\rm max}^2>0.
$$
\end{rem}
We only focus here on the right-going wave, and we show on Figure \ref{test1} the corresponding profile $h^0$, together with the usual profile for the classical Green-Naghdi equations. We also compute the propagation of this solitary wave on the time interval $]0,T]$, with $T=50\,s$ and show the corresponding profiles at several locations along the channel. For these pictures, we set $\delta x=0.125\,m$.
Performing a numerical convergence analysis, we compute the $L^2$ errors at $T=3\,s$ for $\zeta$ and $h\bv$ on a sequence of refined meshes and we observe mean orders of convergence, obtained with linear regressions, of $3.26$ and $3.3$ respectively for $\zeta$ and $h\bv$. \\

\subsubsection{General vorticity model - Case $E_\infty>0$, $v_\infty^\sharp>0$, $F_\infty=0$}

In this second case, we consider the propagation of a solitary wave for non-constant vorticity in the case  $F_\infty=0$. The third order term $F$ remains uniformly equal to $0$ during the propagation, and we can therefore consider the reduced model \eqref{modeleFzero}. The existence of solitary waves solutions is still ruled by Proposition \ref{propSolWav}. There are $2$ solitary waves of opposite velocities given by \eqref{eq:speed_sol1}, but with different profiles. We choose here $E_\infty=0.2\, m^3.s^{-2}$, $v_\infty^\sharp=1.6\, m.s^{-1}$ and again, we set successively $\eps=0.2$ and $\eps=0.3$. Such choices ensure the existence of both left and right-going waves, leading to the profiles shown on the left of Figure \ref{test18} and Figure \ref{test19} respectively for the right-going and the left-going waves in the case $\eps=0.3$. The right-going wave is initially centred at $x=20\,m$ and the left-going wave at $x=170\,m$. Again, we also plot the classical solitary wave profile of same amplitude for comparison purpose. The time evolutions of both right and left going waves are shown on the right of Figure \ref{test18} and Figure \ref{test19}. For these computations, we have set $\delta x=0.07\,m$. We observe the preservation of the initial profiles, together with a very low numerical dissipation for the considered time interval.\\
Again, we perform a numerical convergence analysis on a sequence of refined meshes for both cases. In the case $\eps=0.2$, we observe some convergence rates of $2.9$ for $h\bv$ and $2.88$ for $\zeta, \tE$ and $\dv$. Similar results are obtained in the case $\eps=0.3$, with slightly decreased rates. We also obtain a very similar behavior for the left going wave


\begin{figure}
\centering
\centering
\includegraphics[width=0.45\textwidth]{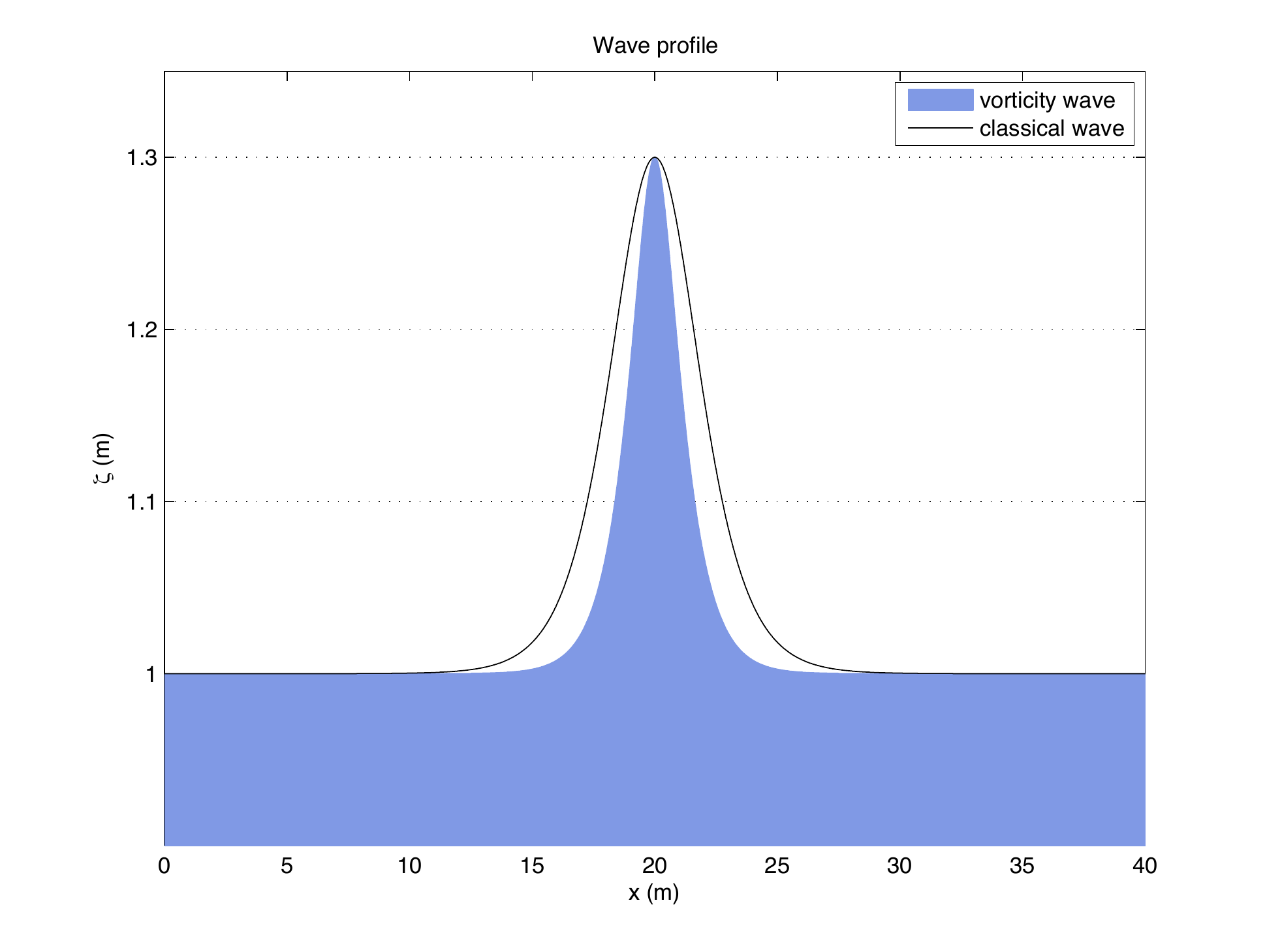}
\includegraphics[width=0.45\textwidth]{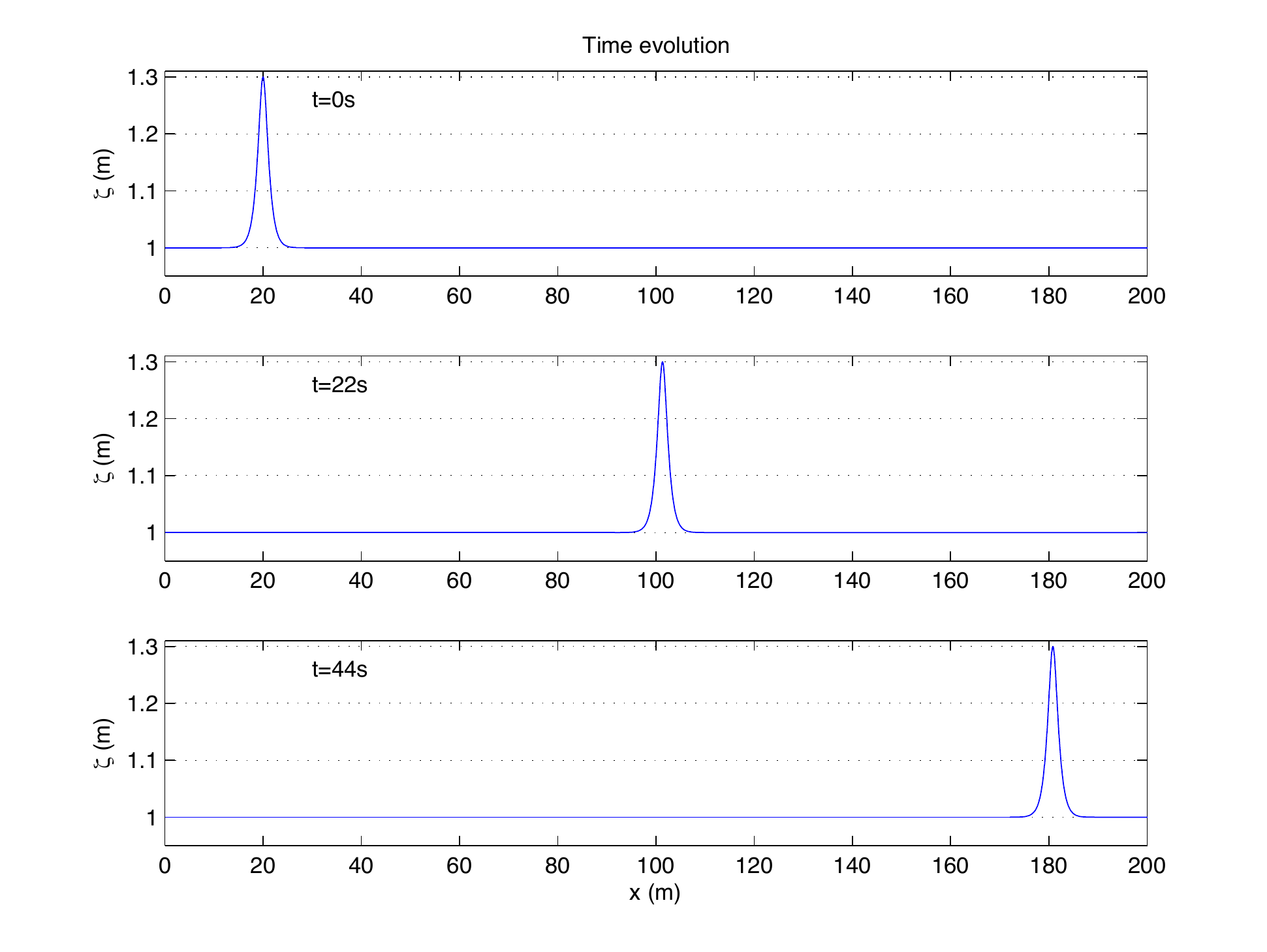}
\caption{Case $E_\infty>0$, $v_\infty>0$, $F_\infty=0$ - Right-going wave with $\eps=0.3$:  initial profile (on the left) and snapshots at different locations along the channel (on the right).}\label{test18}
\end{figure}

\begin{figure}
\centering
\includegraphics[width=0.45\textwidth]{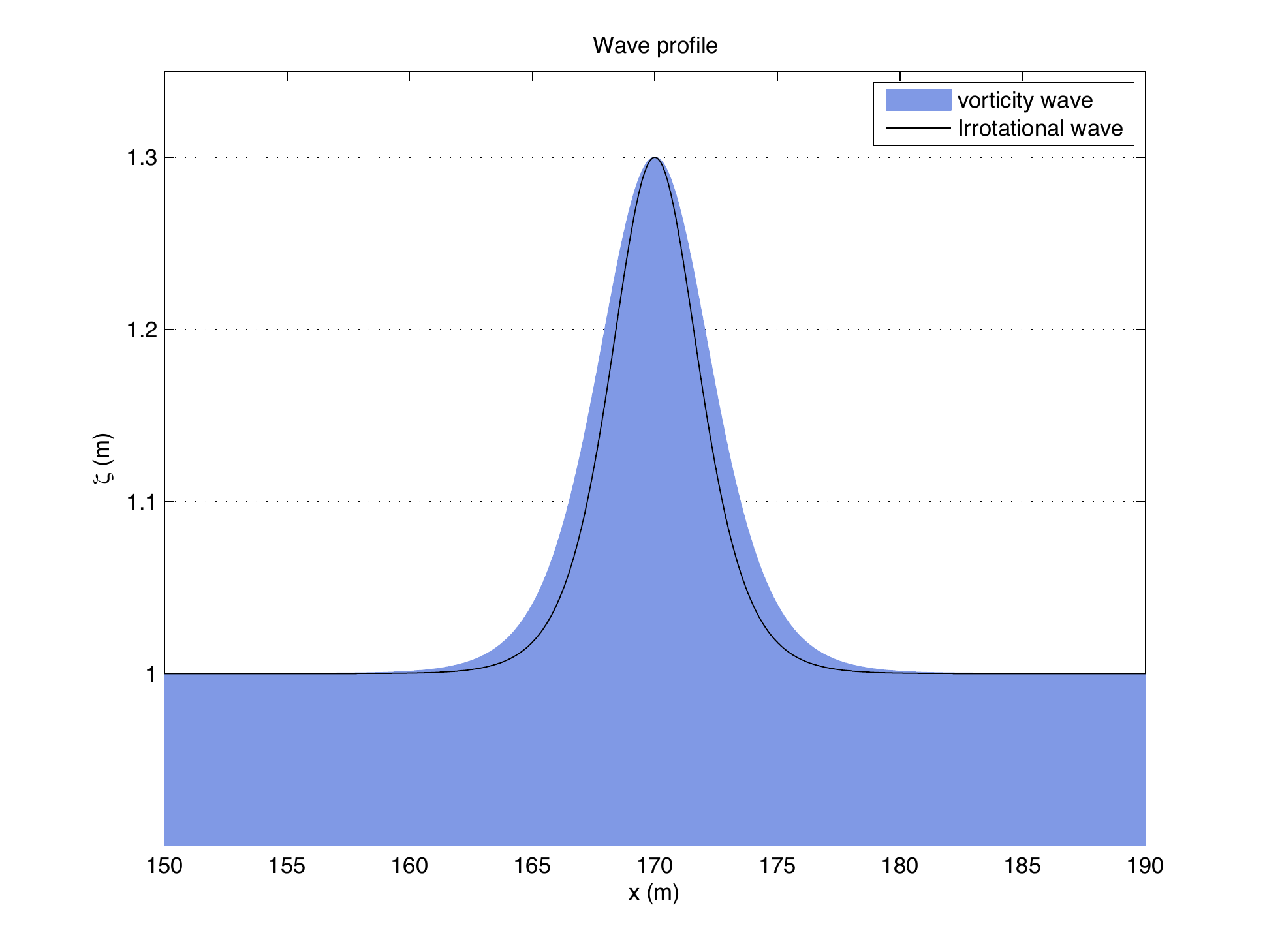}
\includegraphics[width=0.45\textwidth]{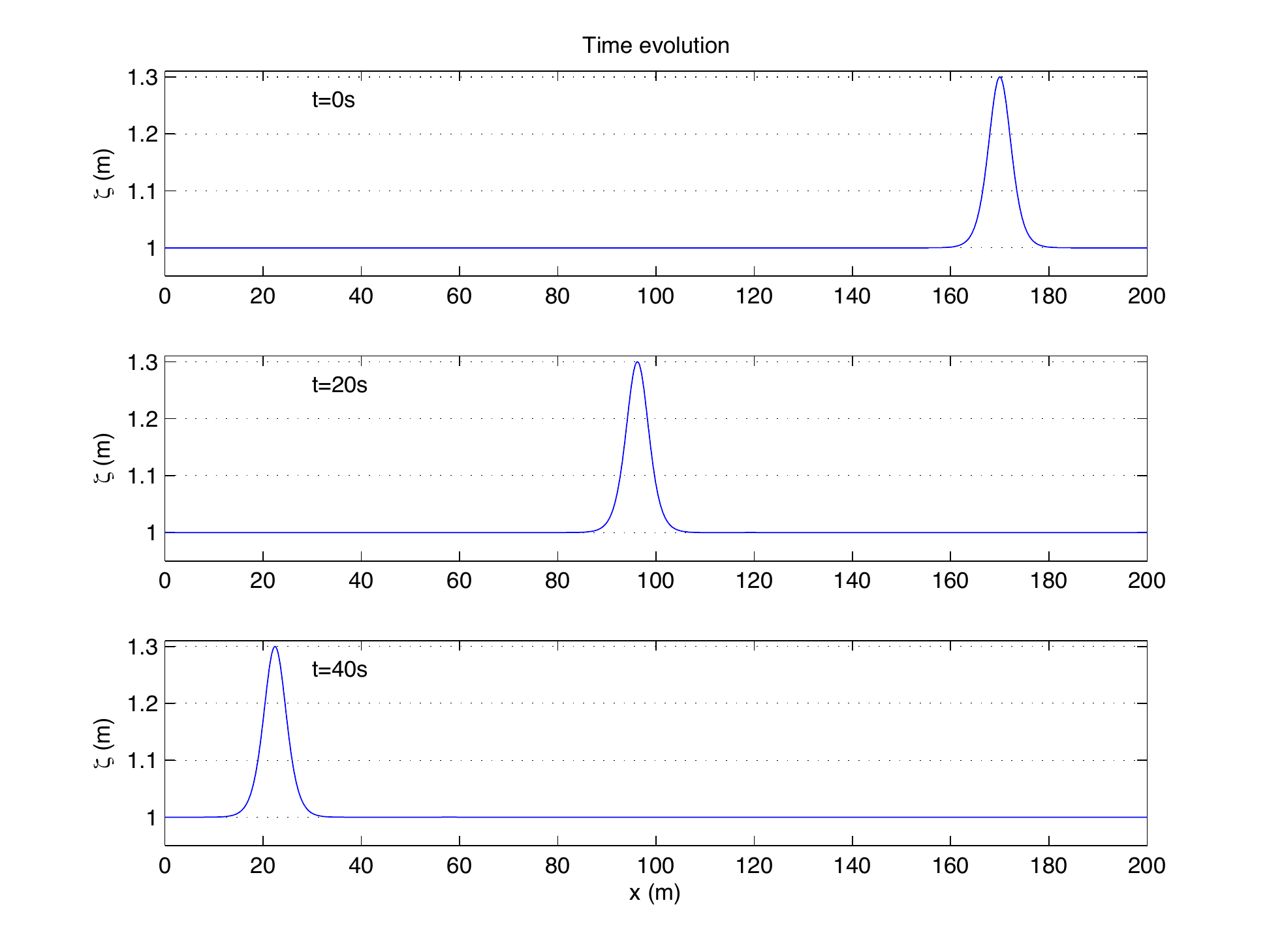}
\caption{Case $E_\infty>0$, $v_\infty>0$, $F_\infty=0$ - Left-going wave with $\eps=0.3$: initial profile (on the left) and snapshots at different locations along the channel (on the right).}
\label{test19}
\end{figure}

\subsubsection{General vorticity model - Case $E_\infty>0$, $v_\infty^\sharp>0$, $F_\infty>0$}

We now focus on the case $F_\infty>0$ and consider the \textit{general vorticity} model \eqref{gn:vort:dim}. The existence of solitary wave solutions is now ruled by Proposition \ref{proposition:2} and we have existence of $2$ solitary waves, provided that conditions \eqref{comp:p2} are fulfilled. We also recall that the waves speeds are now obtained as the minimum and maximum roots of the $3^{rd}$ order polynomial \eqref{eqc}. We consider the following set of parameters $E_\infty=1/12\,m^3.s^{-2}$, $v_\infty^\sharp=1\,m.s^{-1}$, $F_\infty=1/12\,m^4.s^{-3}$ and $\eps=0.5$, ensuring the existence of both left and right going waves and leading to the following wave speeds:
 $$\underbar{c}^- \approx -3.852\,m.s^{-1}\;\; \mbox{and} \;\;\underbar{c}^+\approx 3.93\,m.s^{-1}.
 $$
The corresponding profiles for $\zeta$ are shown on Figure \ref{test13} and Figure \ref{test20}, together with some snapshots of the corresponding time evolutions. For these computations, we set $\delta x=0.05\,m$. The corresponding convergence analysis is performed and we obtain convergence rates of $2.72$ for both $\dv$, $\tF$, $\tE$ and $\zeta$ and $2.78$ for $h\bv$.  We obtain the same behavior for the left going wave.\\

\begin{figure}
\centering
\includegraphics[width=0.45\textwidth]{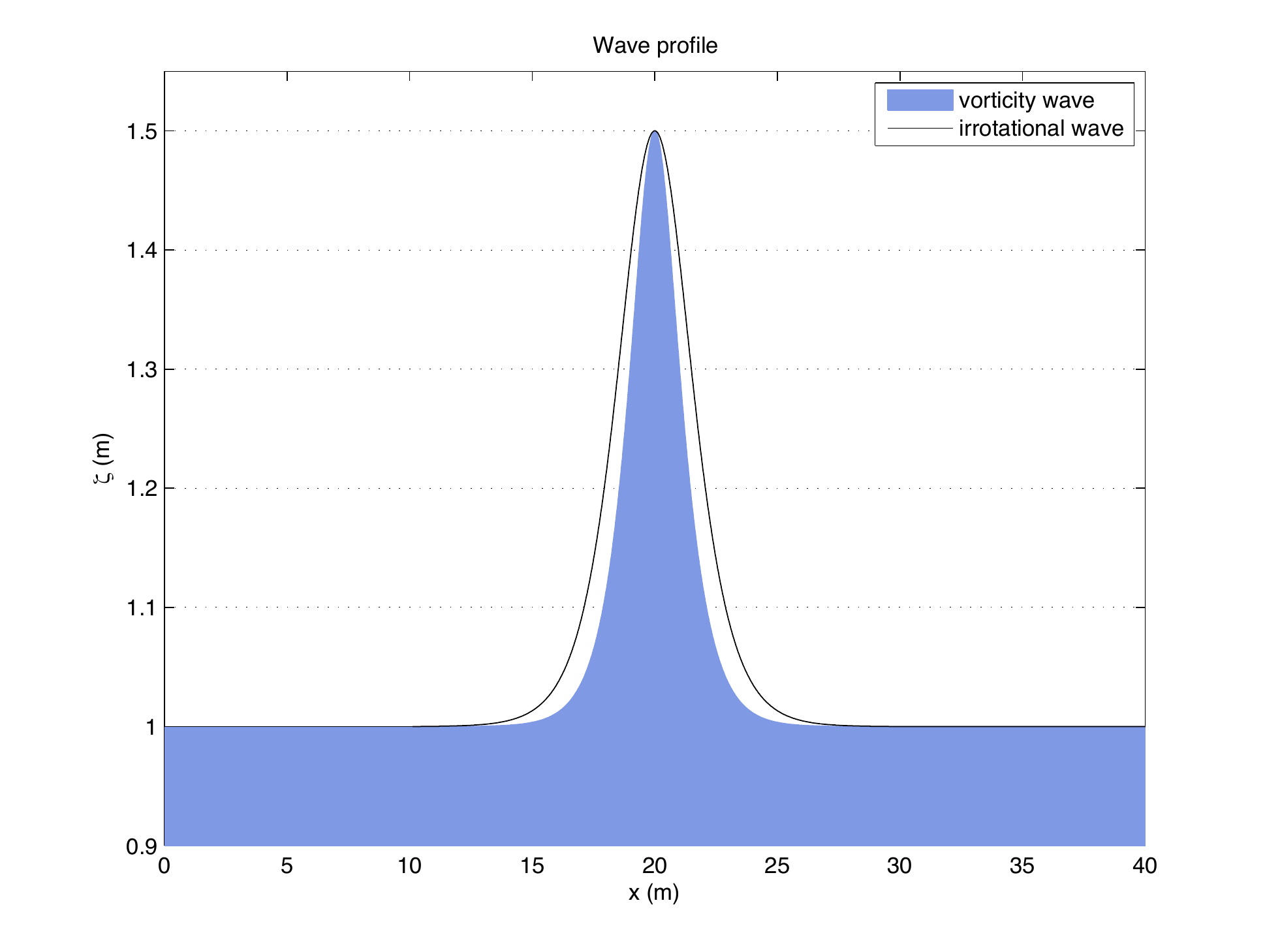}
\includegraphics[width=0.45\textwidth]{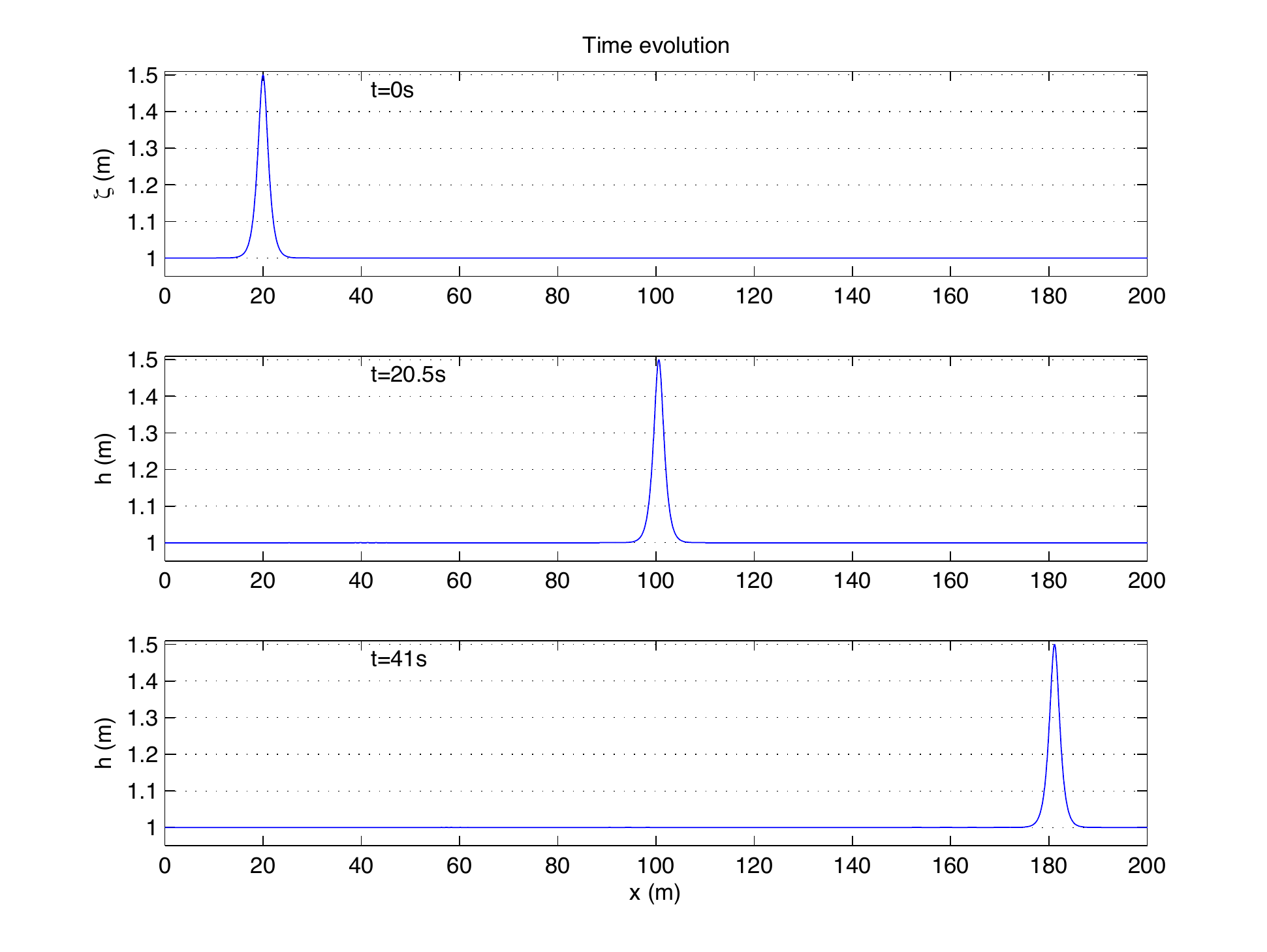}
\caption{Case $E_\infty>0$, $v_\infty>0$, $F_\infty>0$ - Right-going wave with $\eps=0.5$. Initial profile (on the left) and snapshots at different locations along the channel (on the right).}
\label{test13}
\end{figure}

\begin{figure}
\centering
\includegraphics[width=0.45\textwidth]{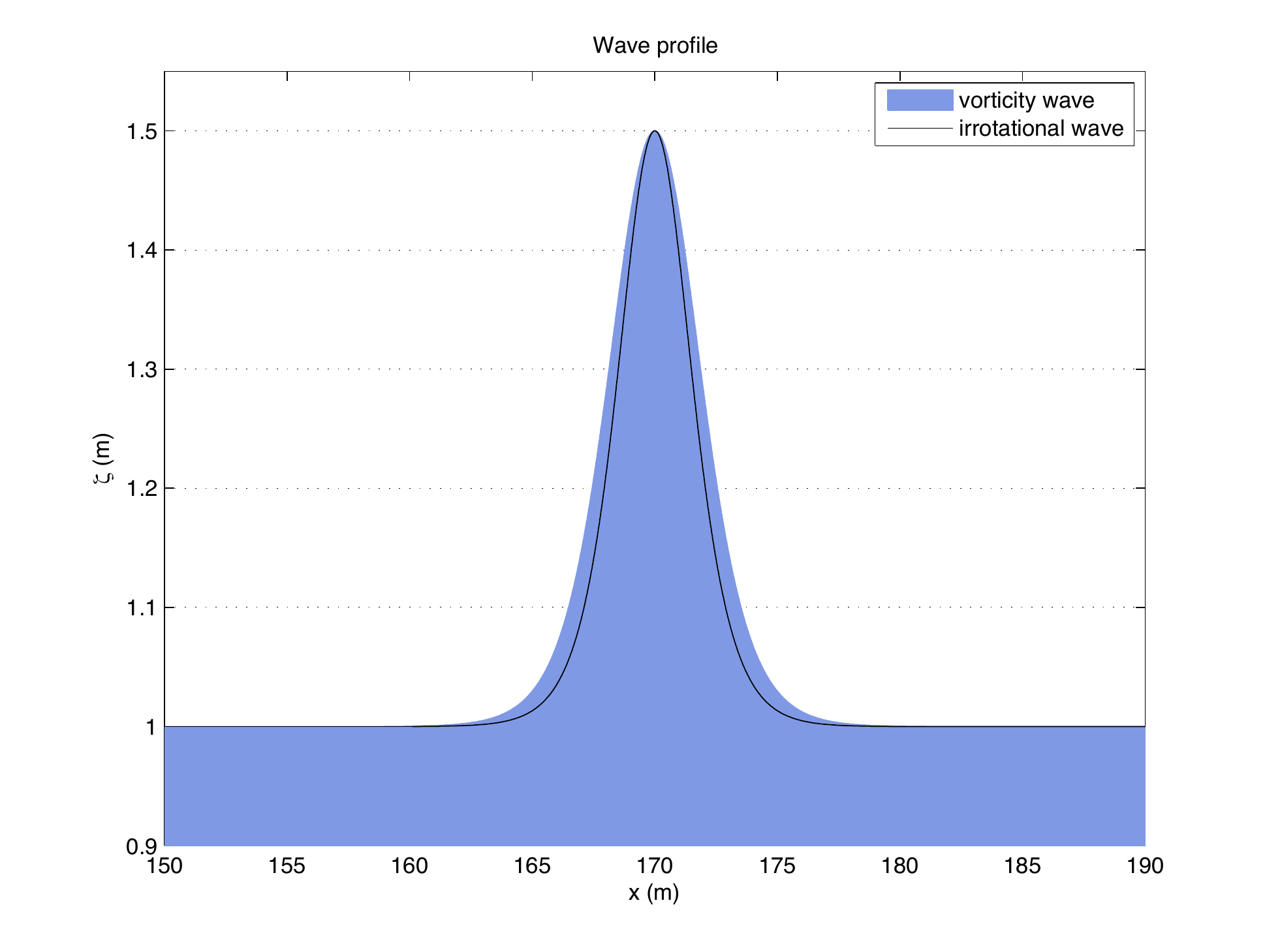}
\includegraphics[width=0.45\textwidth]{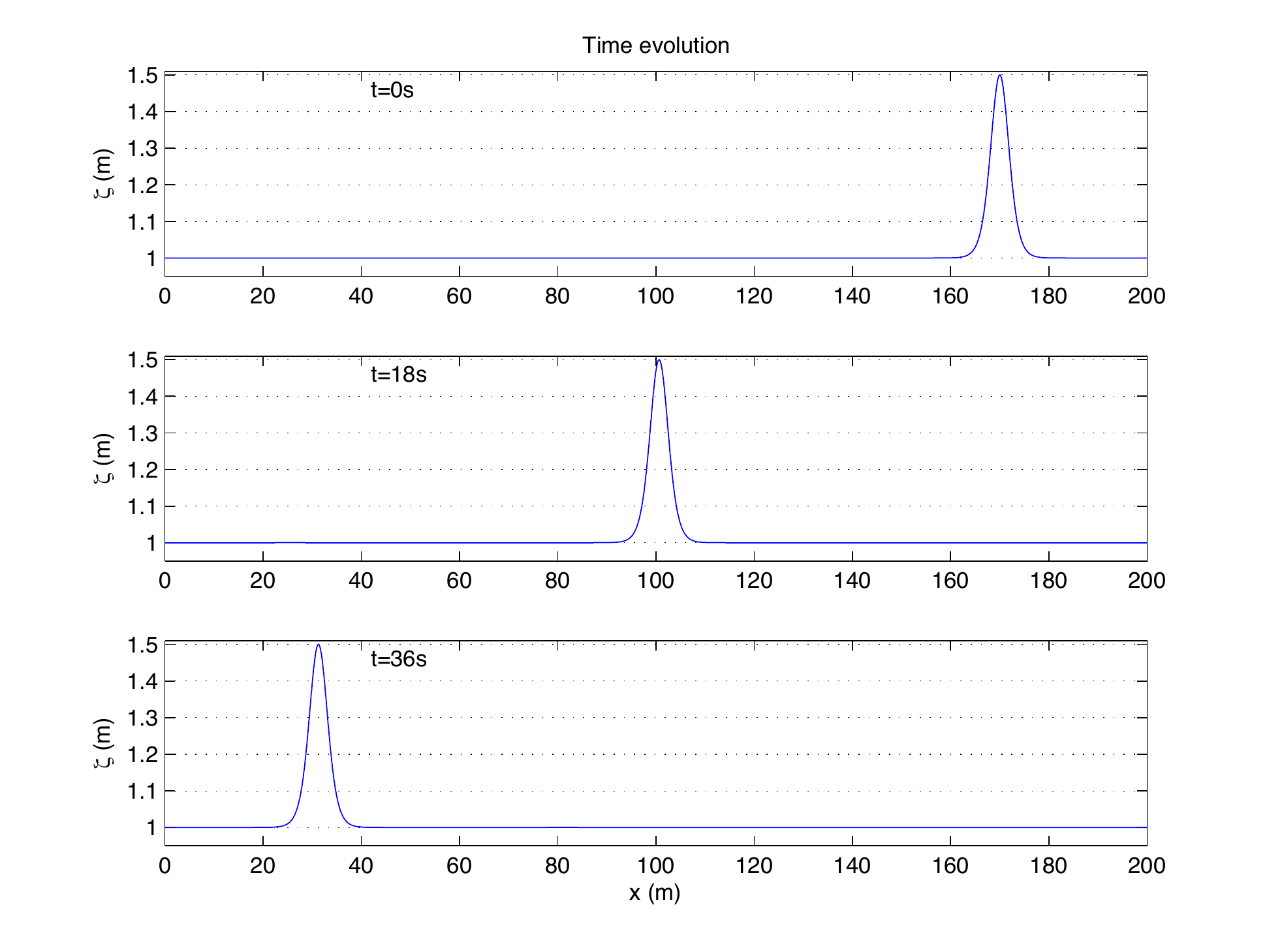}
\caption{Case $E_\infty>0$, $v_\infty>0$, $F_\infty>0$ - Left-going wave with $\eps=0.5$:  initial profile (on the left) and snapshots at different locations along the channel (on the right).}
\label{test20}
\end{figure}

\begin{rem}
As observed in the previous test cases, convergence rates classically seem to decrease with respect to the wave amplitude. Also, note that when $F_\infty\neq 0$, we observe a slight increase of the $L^2$-errors for $\tE$ when compared with the values obtained for $\dv$, $\tF$  and $\zeta$, and the corresponding convergence rates are slightly deteriorated. This phenomenon may be related to the naive finite-difference approximation of the nonlinear and non-conservative terms occurring in the corresponding source terms. 
\end{rem}

\subsection{Peaked solitary waves}\label{sect:PSW}
Let us now briefly illustrate the propagation of \textit{peaked} solitary waves, exhibited in Proposition \ref{prop:peaked}. The numerical simulation of such waves with singularities is a difficult problem as the solution is not smooth. Indeed,  the damping introduced by the numerical methods rapidly smooth out the peak and thereby deform and delay the wave. These issues may be avoided by locally refining the mesh and/or increasing the scheme's accuracy in the vicinity of the wave’s crest. The development of a more involved discrete formulation allowing for $h/p$-adaptivity, based on the recently introduced dG method for the GN equations \cite{DuranMarche:2014ab}, is left for future work.\\

To achieve this, we consider the simplified case $F_\infty=0$, neglecting the third order tensor $F$, and still set $H_0=1\,m$. Choosing the following values  $v_\infty^\sharp=3\,m.s^{-1}$ and $E_\infty=9/12\,m^3.s^{-2}$, we set the relative amplitude in order to be very close to the critical case. In this particular configuration, we have existence of a right-propagating wave only for $h_{max}<h_{crit} = 1.1045\,m$. Consequently, we consider here a wave of amplitude $0.1044\,m$, initially centered at $x=11\,m$. We show the corresponding initial profile on the left side of Figure \ref{test25}. The corresponding propagation is shown on the right side. For this simulation, in order to minimize the numerical diffusion in the vicinity of the singularity we increase the numerical resolution, setting $\delta x=2.5\,10^{-3}\, m$ and use WENO5 reconstructions with a SSP-RK4 time marching scheme.   
\begin{figure}
\centering
\includegraphics[width=0.45\textwidth]{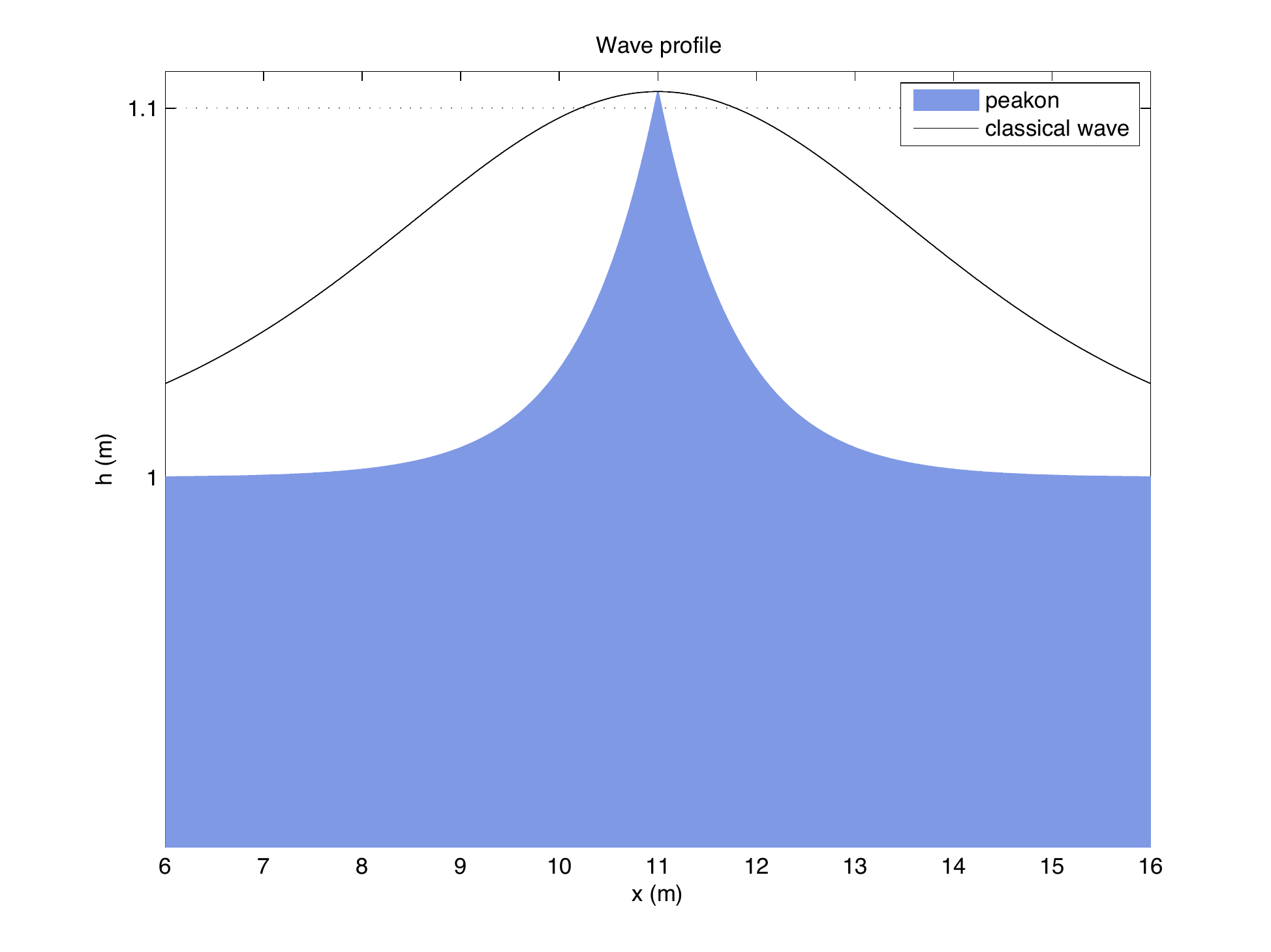}
\includegraphics[width=0.45\textwidth]{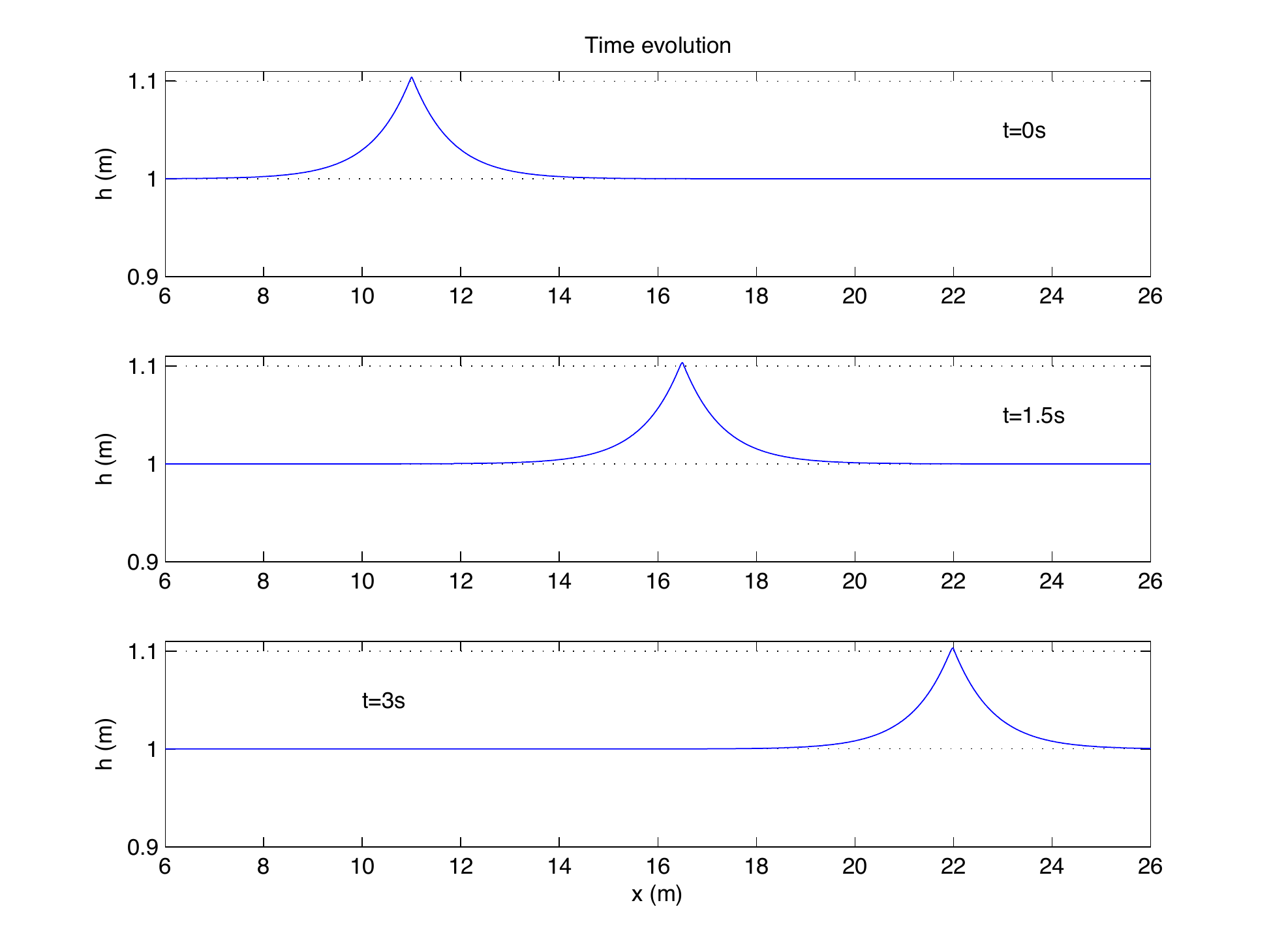}
\caption{Case $E_\infty>0$, $v_\infty>0$, $F_\infty=0$ - Right-going peaked wave with $\eps=0.1044$. We also show in black solid line the shape of the solitary wave for the irrotational Green-Naghdi equations.}
\label{test25}
\end{figure}

\subsection{Influence of vorticity on wave shoaling}\label{sect:vortshoal}

In the following test case, we assess the topography terms discretization, and aim at giving a brief insight into the important study of the impact of the vorticity on wave shoaling. We still consider a $200\,m$ channel with $H_0=1\,m$ but now, we the topography is varying and defined as follows:
$$
b(x)=\frac{H_0}{10}(1+\tanh\big(\frac{x-x_1}{\lambda}\big)),
$$
with $\lambda = 20\,m$, and $x_1=100\,m$. We follow successively the propagation of $3$ solitary waves over this uneven bottom. The first one is a classical solitary wave ($E_\infty = \dv_\infty = F_\infty = 0$) associated with the irrotational Green-Naghdi equations.
In this case, the model \eqref{gn:vorti:dim:2} and the associated numerical approach reduces to the extensively validated framework of \cite{Bonneton20111479, MR2811693, rchgth89}. This wave may therefore be used as a reference solution to highlight the influence of the additional vorticity terms.\\
The second and third ones are solitary wave solutions of \eqref{gn:vorti:dim:2} defined with $F_\infty = 0$ and $(E_\infty, v^\sharp_\infty)$ respectively set to $(8.33\,10^{-2}, -1)$ and $(0.18, -1.5)$. The corresponding celerities are given by $c_0\approx 3.43\,m.s^{-1}$ for the irrotational wave and respectively $c_0\approx 3.47\,m.s^{-1}$ and $c_0\approx 3.53\,m.s^{-1}$ for the second and third vorticity waves. The corresponding initial profiles are shown on Figure \ref{shoaling_1}-left.
\begin{figure}
\centering
\includegraphics[width=0.9\textwidth]{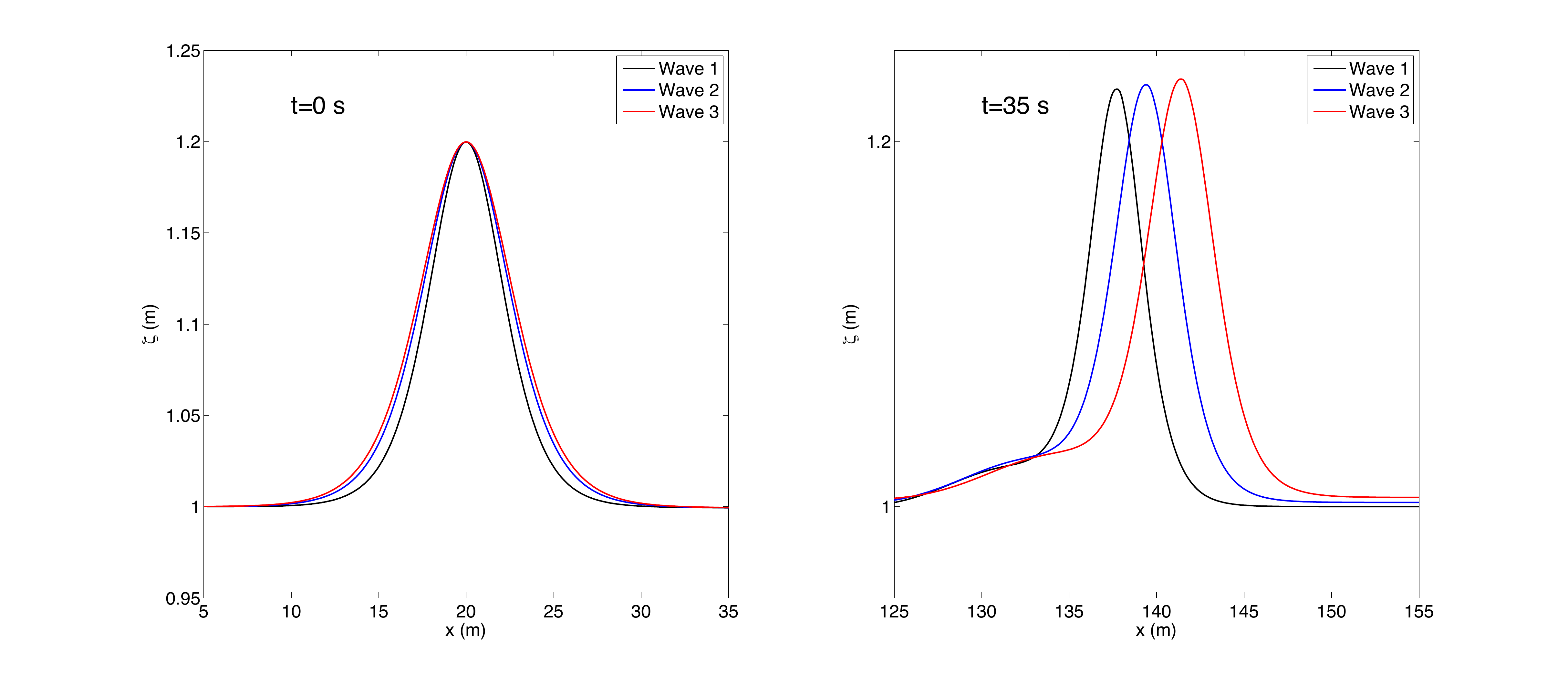}
\caption{Influence of vorticity on wave shoaling: initial wave profiles ($t=0\,s$) on the left and  propagated profiles at $t=35\,s$ on the right. In the legend, 'Wave 1' refers to the irrotational solitary wave, while 'Wave 2' and 'Wave 3' respectively refers to the cases $(E_\infty=8.33\,10^{-2}, \dv_\infty = -1)$ and $(E_\infty=0.18, \dv_\infty = -1.5)$.}
\label{shoaling_1}
\end{figure}
We also show on Figure \ref{shoaling_2} some snapshots of the waves evolution along the channel, and particularly as the waves propagates over the smooth step located in the vicinity of $x_1$. The computation is performed with $\delta x = 0.1\,m$. We can observe the increasing influence of the vorticity on the waves shoaling processes, as the waves amplitude increase with respect to the vorticity magnitude, as well as on the stretching area at the rear of the waves and on the mean level at the front. A zoom on the profiles at $t=35\,s$ is shown on Figure \ref{shoaling_2}-right. Of course, extensive studies are still needed to systematically analyze and accurately quantify this influence. 
\begin{figure}
\centering
\includegraphics[width=0.8\textwidth]{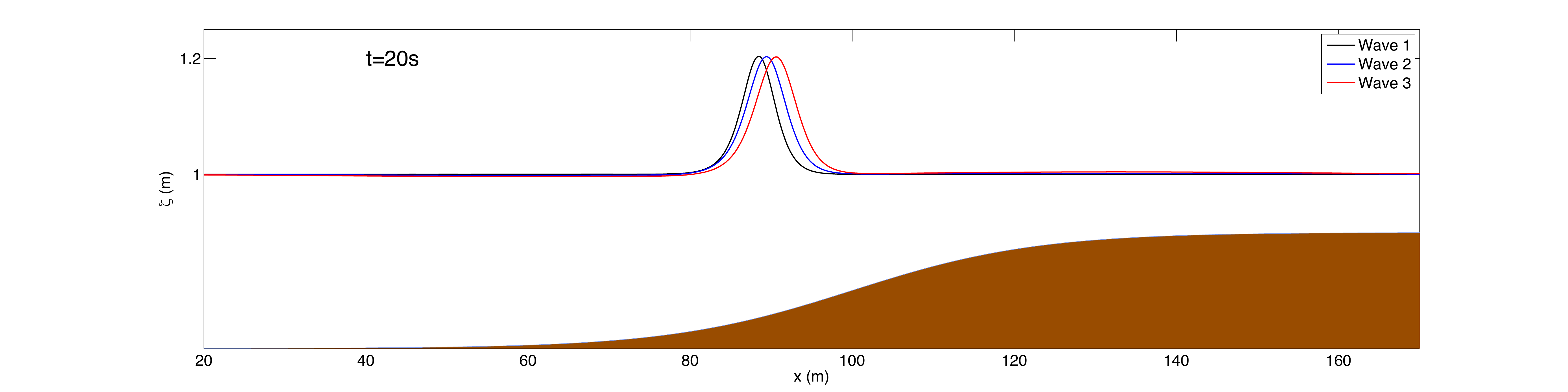}
\includegraphics[width=0.8\textwidth]{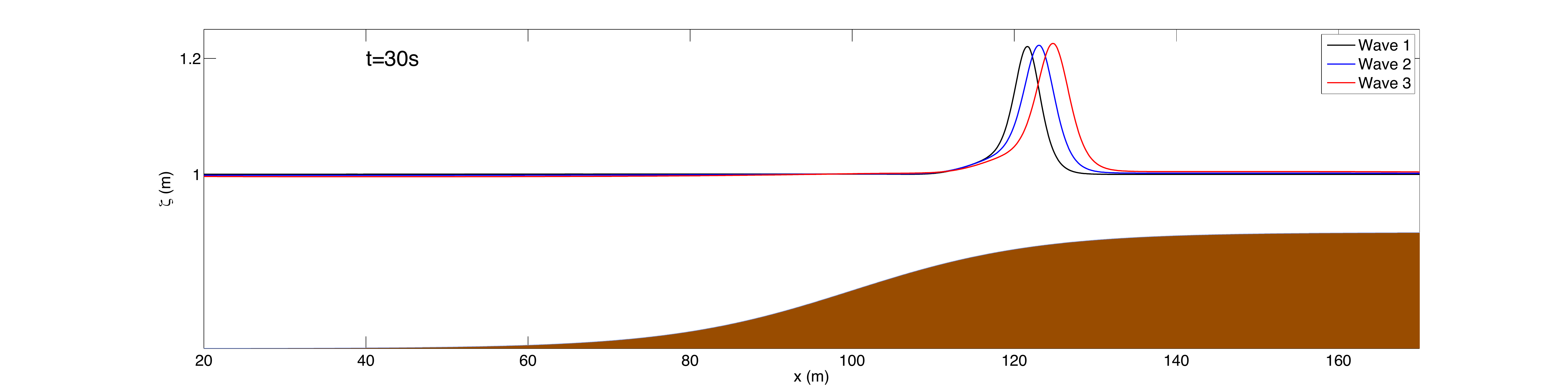}
\includegraphics[width=0.8\textwidth]{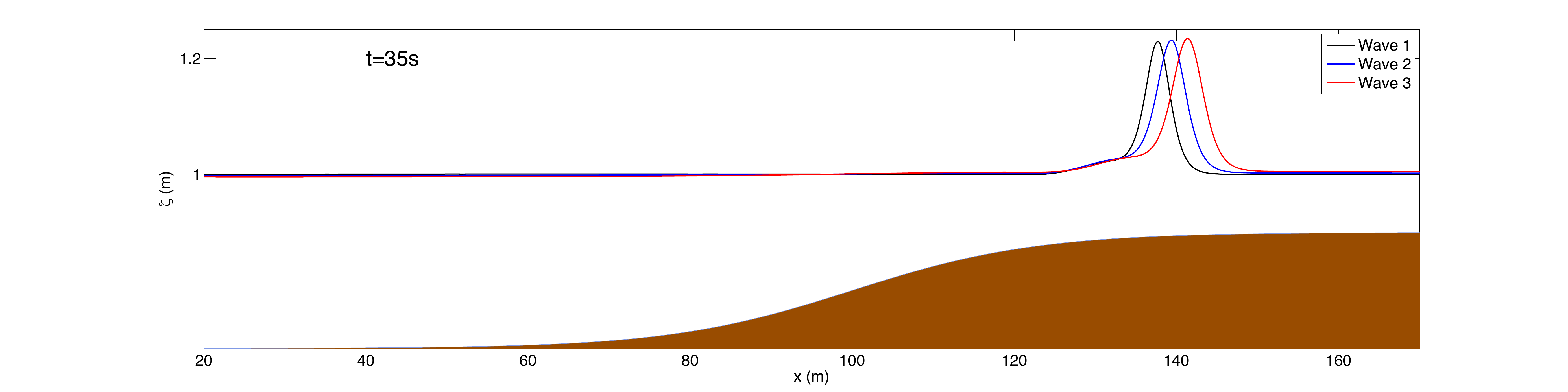}
\caption{Influence of vorticity on wave shoaling: evolution of the waves profiles along the channel at times $t=20, 30$ and $35\,s$.}
\label{shoaling_2}
\end{figure}

\section{Evolution of the velocity field inside the fluid domain}\label{sec:velocity}

In the previous numerical computations we computed the surface
elevation $\zeta$, the average velocity $\ovv$ as well as $v^\sharp$,
$E$ and $F$ from their initial value through the resolution of the {\it
  one}-dimensional equations
(\ref{gn:vort:dim}). We show here how this system of equations can be
used to describe the dynamics of the $(1+1)$-dimensional velocity
field $U=(v, w)$ and of the surface elevation in terms of their initial value
$U^0$ and $\zeta^0$. 
\subsection{The initial condition}
Denoting by $\ovv^0$ the vertical average of the horizontal velocity  at $t=0$ and by $v^{*,0}$ its fluctuation, one can write
\begin{equation}\label{initvw}
\begin{cases}
v^0=&\ovv^0+v^{\star,0},\\
w^0=&- \dx\big((H_0+z-b)\ovv^0\big)-\dx \int_{-H_0+
  b}^z v^{\star,0};
\end{cases}
\end{equation}
the expression for $w^0$ is deduced from the expression for $v^0$ thanks
to the incompressibility condition and the non penetration condition
at the bottom. We also recall that the initial vorticity $\omega^0$ is defined as
$$
\omega^0=\dz u^0-\dx w^0.
$$
We shall consider here initial conditions that correspond to shallow water configurations (i.e. $\mu=H_0^2/L^2\ll1$). Denoting by 
$v^{0}_{\rm sh}$ the shear velocity induced by the vorticity field $\omega^0$ and by $v^{\star,0}_{\rm sh}$ its fluctuation around its vertical mean value,
\be\label{v:shear}
v^{\star,0}_{\rm sh}=-\int_{z}^{\zeta^0}\omega^0+\frac{1}{h^0}\int_{-H_0+b}^{\zeta^0}
\int_{z}^{\zeta^0} \omega^0\qquad \qquad(h^0=H_0+\zeta^0-b),
\ee
it is shown  in \cite{castroLannes:2014aa}, Eq. (2.30), that up to $O(\mu^{3/2})$ terms, one has
 \begin{equation}\label{vstar:shear}
 v^{*,0}=v^{\star,0}_{\rm sh}+T^*\ovv^0
 \end{equation}
where $T^*\ovv^0$ accounts for the fluctuations due to dispersion and is given by
$$
T^\star\ovv^0=-\frac{1}{2}\big((z+H_0-
b)^2-\frac{1}{3}(h^0)^2\big)\dx^2\ovv+\big(z-\zeta^0+\frac{1}{2}h^0\big)\big(\dx
b\dx\ovv^0+\dx(v^0\dx
b)\big).
$$
Since our main goal here is to comment on the effects of the vorticity on the propagation of waves, it is convenient to consider an initial vorticity field $\omega^0$ and to construct the corresponding approximate velocity field $U^0$ given by \eqref{initvw}, \eqref{v:shear} and \eqref{vstar:shear}.

\subsection{Reconstruction methodology}

 Let us now briefly recall the procedure described in
\cite{castroLannes:2014aa} to recover $U$ at all times from $\zeta^0$ and
$U^0$. It is of course sufficient to compute $U$ on each level line
$\theta\in [0,1]$, that is, 
\begin{align*}
v_\theta(t,x)&=v\big(t,x,-H_0+b(x)+\theta h(t,x)\big),\\
w_\theta(t,x)&=w\big(t,x,-H_0+b(x)+\theta h(t,x)\big).
\end{align*}
\begin{rem}
One can also define the vorticity on the level lines by
$\omega_\theta(t,x)=\omega\big(t,x,-H_0+b(x)+\theta h(t,x)\big)$. It can be
computed in terms of $v_\theta$ and $w_\theta$ by the formula
$$
\omega_\theta=\frac{1}{h}\partial_\theta v_\theta-\dx
w_\theta+\frac{\dx (b+\theta h)}{h}\partial_\theta w_\theta.
$$
\end{rem}
It is shown in \cite{castroLannes:2014aa} that the velocity fields conserves
its initial structure \eqref{initvw}-\eqref{vstar:shear} and that one has
\begin{align}\label{comp:vtheta}
v_\theta&=\ovv+v^\star_{{\rm sh},\theta}+T^\star_\theta \ovv\\
\label{comp:wtheta}
w_\theta&=-\dx \big( h(\theta \ovv+Q_\theta+\int_0^\theta T^\star_{\theta'}\ovv
d\theta')\big)+\dx(-H_0+b+\theta h)v_\theta.
\end{align}
with
$$
T^\star_\theta \ovv=-\frac{1}{2}(\theta^2-\frac{1}{3})h^2
\dx^2\ovv+(\theta-\frac{1}{2})h\big(\dx b\dx \ovv+\dx (\dx b\ovv)\big).
$$

The computation of the different quantities involved in these expressions must be performed as follows:
\begin{enumerate}
\item From the initial datas $\zeta^0$, $\ovv^0$ and $v_{\rm
    sh}^{\star,0}$ compute the initial values $E^0$, $F^0$ and
  $v^{\sharp,0}$ from their definitions \eqref{defvsharp},
  \eqref{defE} and \eqref{defF}.
\item Compute $\zeta$, $\ovv$, $v^\sharp$, $E$ and $F$ on the time
  interval $[0,T]$ by solving the Green-Naghdi equations
  \eqref{gn:vort:dim}. 
\item Compute the quantities $q_\theta$ and $Q_\theta$ on the same time
  interval by solving
$$
\dt q_\theta+\dx(\ovv q_\theta)=0,\qquad \dt Q_\theta+\dx(\ovv Q_\theta)=0
$$
with initial conditions $q_\theta^0=\partial_\theta v_{{\rm sh},\theta}^{\star,0}$
and $Q_\theta^0(x)=\int_0^\theta v_{{\rm sh},\theta'}^{\star,0}d\theta'$, where we use
the notation $v^{\star,0}_{{\rm sh},\theta}(x)=v^{\star,0}_{\rm sh}\big(t,x,-H_0+b(x)+\theta h^0(t,x)\big)$.
\item Compute $v_{{\rm sh},\theta}^\star$ by solving on $[0,T]$ the equation
$$
\dt
v^\star_{{\rm sh},\theta}+\dx\big(v_{{\rm sh},\theta}^\star(\ovv+\frac{1}{2}v_{{\rm sh},\theta}^\star)\big)=\frac{1}{h}\dx
E+\frac{q_\theta}{h}\dx (hQ_\theta),
$$
with initial data $v^{\star,0}_{{\rm sh},\theta}$. 
\item Use \eqref{comp:vtheta} and \eqref{comp:wtheta} to get
      $v_\theta$ and $w_\theta$.
\end{enumerate}

\subsection{Implementation hints}

From a practical point of view, we aim at computing the velocity field $U$ on some given $N_z$ level lines along the vertical layer of fluid. We introduce the discrete increments
$$
\theta_j = \frac{\delta z + (j-1)\delta z}{H_0},\;\;1\leq j\leq N_z,\quad\quad\mbox{with}\quad \delta z = \frac{H_0}{N_z},
$$
such that, for some given time $t$ and horizontal coordinate $x$, the $j^{th}$ level line is located at the vertical coordinate $z=-H_0+b(x)+\theta_j h(t,x)$. 
To compute the velocity field on each of these level lines, we supplement the Green-Naghdi model with general vorticity \eqref{gn:vorti:dim:2} with a set of equations describing the time evolution of the auxiliary quantities $q_{\theta_j}$, $Q_{\theta_j}$ and $v_{\theta_j}^\star$:
\be\label{gn:vort:dim:vel:1}
\begin{cases}
&\dt h + \dx(h\bar{v}) = 0,\\
&(1+\mathbb{T})[\dt(h\bv) + \dx(h\bv^2)] + gh\dx \zeta + h\cQ_1(\bv)
+\dx (h^2\tE)+ h\cC(\bv,\dv) =0,\\
&\dt \dv + \dx( \bv\tdv  ) =0,\\
&\dt \tE +  (\bv\tE)_x+ 3\tF \dx h + h\dx\tF = 0,\\
&\dt \tF +  \dx(\bv\tF)= 0,\\
&\dt q_{\theta_j}+\dx(\ovv q_{\theta_j})=0,\quad 1\leq j\leq N_z,\\
& \dt Q_{\theta_j}+\dx(\ovv Q_{\theta_j})=0,\quad 1\leq j\leq N_z,\\
&\dt v^\star_{\theta_j}+\dx\big(v_{\theta_j}^\star(\ovv+\frac{1}{2}v_{\theta_j}^\star)\big)=\dsp\frac{1}{h}\dx(h^2\tE)
+\dsp\frac{q_{\theta_j}}{h}\dx (hQ_{\theta_j}),\quad 1\leq j\leq N_z.
\end{cases}
\ee
Considering the numerical method introduced in \S\ref{sect:num:scheme}, the easiest way to account for these additional equations, from a discrete point of view, is to incorporate the auxiliary quantities $q_{\theta_j}, Q_{\theta_j}$ and $v_{\theta_j}^\star$ into our splitting approach, and into the approximate Riemann solver for the transport part of the equations.  \\
\begin{rem}
As these new equations are decoupled from the equations on mass and momentum, we choose to keep the simple $3$ waves structure of the approximate HLLC Riemann solver.  Such a choice may appear as surprising for the quantities $v_{\theta_j}^\star$, considering the particular form of the associated non-linear evolution equations. However, our numerical investigations have shown that considering each scalar equations separately, the corresponding wave speeds associated with the values of $\ovv+\frac{1}{2}v_{\theta_j}^\star$ at interfaces are very close to $\ovv$ (which is in accordance with the vorticity scaling studied here, see \cite{castroLannes:2014aa}), and always lie between the considered external wave speed estimates. 
\end{rem}

\vspace{0.2cm}
\noindent
To simplify the notations, let us consider the case of one level line, located at the vertical coordinate $z=-H_0+b(x)+\theta h(t,x)$, with $\theta\in ]0,1]$, and denote by $q_{\theta}, Q_{\theta}$ and $v_{\theta}^\star$ the associated quantities. The splitting scheme \eqref{S1:gen}-\eqref{S2:gen} is modified as follows:
\begin{itemize}
\item  $S_1(t)$ is the solution operator
associated to the conservative \textit{propagation} step, which may be written in the following conservative form
\beq\label{syst:compact:theta}
\dt \cW_\theta + \dx \eF(\cW_\theta)  = \eS(\cW_\theta,b),
\eeq
with $\cW_\theta = {}^T(h,\, h\bv,\, \dv,\, \tE,\, \tF,\, q_\theta,\, Q_\theta,\, v_\theta^\star)$,
$$
\eF(\cW_\theta) = 
{}^T(
 h\bv,\,
 h\bv^2 + p(h,\tE),\,
\bv\dv,\,
 \bv \tE,\,
 \bv \tF,\,
 \bv q_\theta,\,
 \bv Q_\theta,\,
 v_\theta^\star(\bv + \frac{1}{2} v_\theta^\star)
),
$$
and $
\eS(\cW_\theta,b) = 
 (0,\,
-gh b_x,\,
0,\,
 0,\,
0,\,
0,\,
0,\,
0).$

\item $S_2(t)$ is the solution operator associated to the \textit{dispersive correction},
\be
	\label{S2:gen:mod}
	\left\lbrace
	\begin{array}{lcl}
	\vspace{0.1cm}
	\dsp \dt h =0,\\
	
	\vspace{0.1cm}
	\dsp \dt (h \bv) - gh \dx\zeta -  \dx(h^2\tE) \\
	\quad\quad\quad+ (1+ \mathbb{T})^{-1}\big[gh\dx\zeta + \dx(h^2\tE) +  h\cQ_1(\bv) + h\cC(\bv,\dv)\big]=0,\\
	
	\vspace{0.1cm}
	\dt\dv  =0,\\
	
	\vspace{0.1cm}
        \dt\tE   + 3\tF \dx h + h\dx \tF= 0,\\
        	\vspace{0.1cm}
        \dt\tF  = 0,\\
        	\vspace{0.1cm}
        \dt q_\theta = 0,\\
        	\vspace{0.1cm}
        \dt Q_\theta = 0,\\
        	\vspace{0.1cm}
        \dt v_\theta^\star -\dsp\frac{1}{h}\dx(h^2\tE)
-\dsp\frac{q_{\theta}}{h}\dx (hQ_{\theta}) = 0.
	\end{array}\right.
\ee
Note that when compared with \eqref{S2:gen}, the second step is only modified with the introduction of a source term in the equation on $v_\theta^\star$.
\end{itemize}

 
\subsection{A solitary wave propagating in a constant vorticity field}

To illustrate the previous reconstructions, we show here the velocity fields associated with some of the new solitary wave solutions exhibited in \S\ref{sect:derSol}. We choose the simplest case of a constant vorticity field, described by model \eqref{gn:const} and we set
$$
\curl \bU=(0,\omega,0)^T\quad \mbox{ with }\quad \omega(t,x,z)=\omega_0.
$$
Using \eqref{v:shear}, we obtain the initial horizontal shear velocity field
$$
v^{\star,0}_{\rm sh}(x,z)= -\omega_0(\zeta^0 -z) + \frac{\omega_0}{2}h^0,
$$
and the initial velocity field $v_\theta^0$ et $w_\theta^0$ are obtained with straightforward computations using \eqref{comp:vtheta} and \eqref{comp:wtheta}. This initial velocity field is shown on Figure \ref{fig2Dvelocity:constant} for the values $\omega_0=0.3\,s^{-1}$ and $\omega_0=1.5\,s^{-1}$. For comparison purpose, we also show the velocity field corresponding to the irrotational solitary wave of similar amplitude. We can observe the vertical decreasing of the velocity's horizontal component magnitude for the rotational solitary wave, while this is left unchanged for the irrotational wave.
\begin{center}
\begin{figure}
\includegraphics[width=0.9\textwidth]{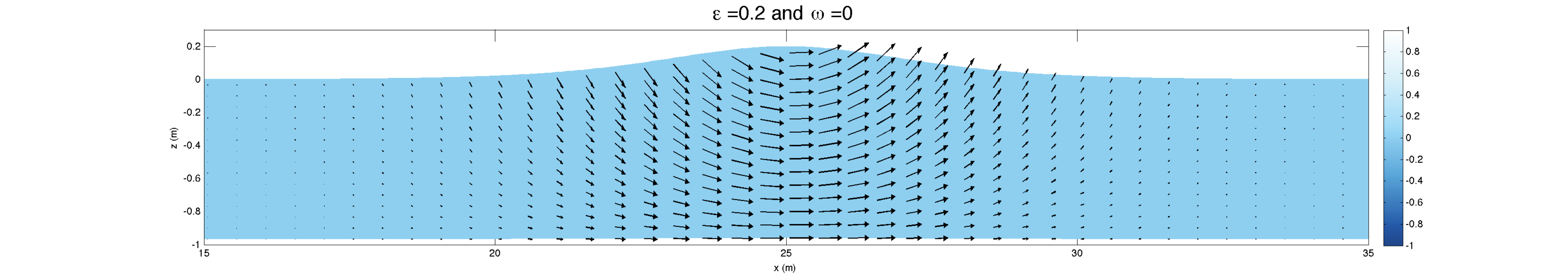}
\includegraphics[width=0.9\textwidth]{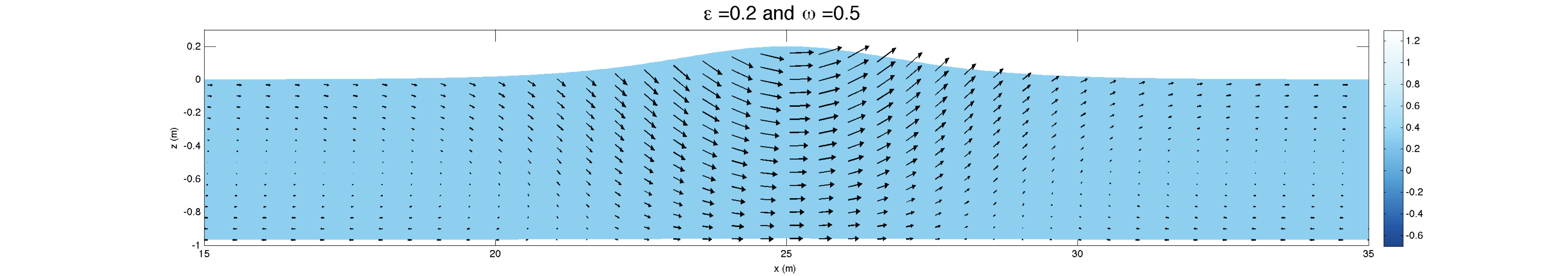}
\includegraphics[width=0.9\textwidth]{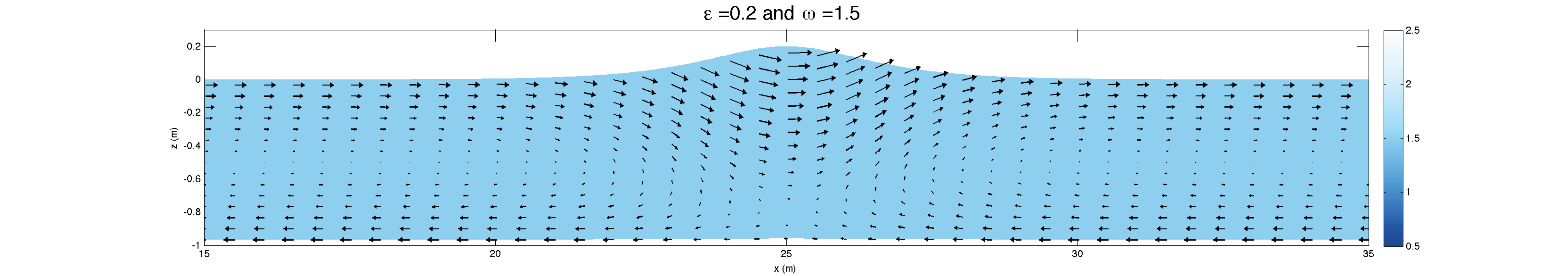}
\caption{Velocity fields in solitary waves with constant vorticity: influence of the vorticity magnitude.}
\label{fig2Dvelocity:constant}
\end{figure}
\end{center}

\subsection{A solitary wave arriving in a vorticity region}\label{sect2Dex1}

We consider now an initial vorticity field with vertical and horizontal dependencies, defined as follows:
\be\label{init:vorti}
\omega^0(x,z)=-2\pi \underline{\omega}\cos\big( 2\pi\frac{\zeta-z}{H_0}\big)\exp\big(- \frac{(x-x_1)^2}{\lambda^2} \big)
\ee
with $\underline{\omega}=\frac{\eps\sqrt{gH_0}}{\lambda}$ (this
scaling of the vorticity corresponds to the regime studied in
\cite{castro:lannes:2014:1, castroLannes:2014aa} where the Green-Naghdi models with
vorticity have been derived and justified). 
\\ 
With this definition, we observe that the vorticity vanishes as $x\to\pm\infty$. Additionally, its strongest variations of amplitude along the horizontal dimension are located near some point of abscissa $x_1$.\\
In the following, we aim at qualitatively observing the behavior of an irrotational solitary wave arriving from afar (where the vorticity is therefore negligible). We consider a channel of $50\,m$ long, set $H_0=1\,m$, $x_1=25\,m$ and take as initial surface elevation the profile of an irrotational solitary wave \eqref{solitary:gn} centred at $x_0=12.5\,m$ with $\eps=0.25$. We highlight that $x_0$ is chosen far enough from $x_1$ to ensure that the vorticity in initially negligible in the vicinity of the solitary wave. The choice \eqref{init:vorti} for the initial vorticity field leads to the following initial horizontal shear velocity
$$
v^{\star,0}_{\rm sh}(x,z)=\underline{\omega}H_0\big[\sin\big(
2\pi\frac{\zeta^0-z}{H_0}\big)
-\frac{H_0}{2\pi}\frac{1-\cos(2\pi\frac{h^0}{H_0})}{h^0}\big]
\exp\big( -\frac{(x-x_1)^2}{\lambda^2} \big)
$$
so that the full initial velocity field can again be
constructed\footnote{The two integrals in the expression for $w^0$ can
be computed explicitly,
\begin{align*}
\int_{-H_0}^z v^{*,0}_{\rm sh}&=\frac{\underline{\omega}H_0^2}{6\pi}
\Big[\cos\big(
2\pi\frac{\zeta-z}{H_0}\big)
- \cos\big(2\pi\frac{h^0}{H_0}\big)
-\frac{z+H_0}{h^0}(1-\cos(2\pi\frac{h^0}{H_0}))\big]
\exp\big( -\frac{(x-x_1)^2}{\lambda^2} \Big)\\
\int_{-H_0}^z T^\star\ovv^0&=-\frac{1}{6}(z+H_0)\big((z+H_0)^2-(h^0)^2\big)\dx^2\ovv^0.
\end{align*}
} through \eqref{initvw} (see Figure \ref{fig2Dvelocity}). We now follow the procedure described in the previous
sections. From simple computations, one gets
\begin{align*}
v^{\sharp,0}&=\frac{2H_0^2 \underline{\omega}}{\pi^3(h^0)^3}
\big[(h^0)^2\pi^2(1+2\ucos^2)-3H_0^2\usin^2\big]\exp(-\frac{(x-x_1)^2}{\lambda^2})\\
E^0&=\frac{3 H_0^2 \underline{\omega}^2}{2\pi^2 h^0}
\big[ H_0 h^0\pi (\usin^2-\ucos^2)
\usin\ucos-2H_0^2\usin^4+\pi^2(h^0)^2\big]\exp(-2\frac{(x-x_1)^2}{\lambda^2})\\
F^0&=\frac{3H_0^4\underline{\omega}^3}{2\pi^3(h^0)^2} \usin^2
\big[ 4\pi^2 (h^0)^2(1-2\ucos^2)\ucos^2-9\pi H_0h^0
(1-2\ucos^2)\ucos\usin\\
&\hspace{3cm}+12 H_0^2(1-\ucos^2)^2-5\pi^2(h^0)^2\big]\exp(-3\frac{(x-x_1)^2}{\lambda^2})
\end{align*}
where we used the notations $\usin=\sin(\pi \frac{\zeta^0}{H_0})$ and
$\ucos=\cos(\pi \frac{\zeta^0}{H_0})$. These functions are shown
on Figure \ref{figEFV}. Note in particular that the function $F$ 
almost identically vanishes (its sup norm does not exceed $2.10^{-9}$), and
it is therefore possible to work with the reduced model \eqref{modeleFzero}. For the computation, we set $N_x=500$, leading to $\delta x= 0.1\,m$, and $N_z=30$. The corresponding results are shown on Figure \ref{fig2Dvelocity}, where we plot the free surfaces at several times during the propagation, together with the reconstructed velocity fields and the corresponding vorticity. We can observe the impact of the current on the initial velocity field and the free surface profile, and that the initial vorticity area is slightly dragged along and stretched as the wave propagates. 

\begin{figure}
\centering
\includegraphics[width=0.9\textwidth]{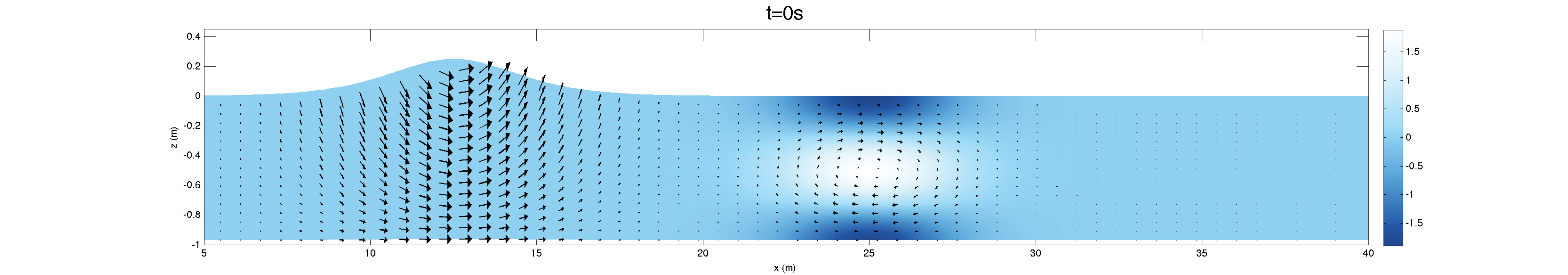}
\includegraphics[width=0.9\textwidth]{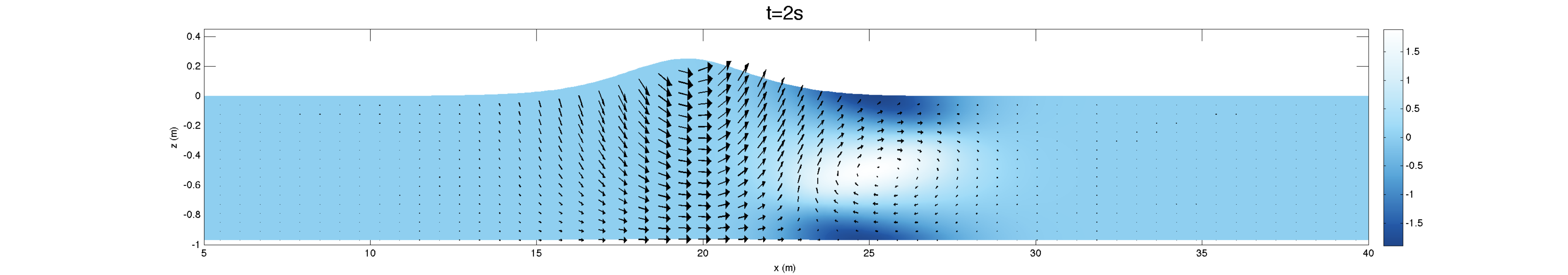}
\includegraphics[width=0.9\textwidth]{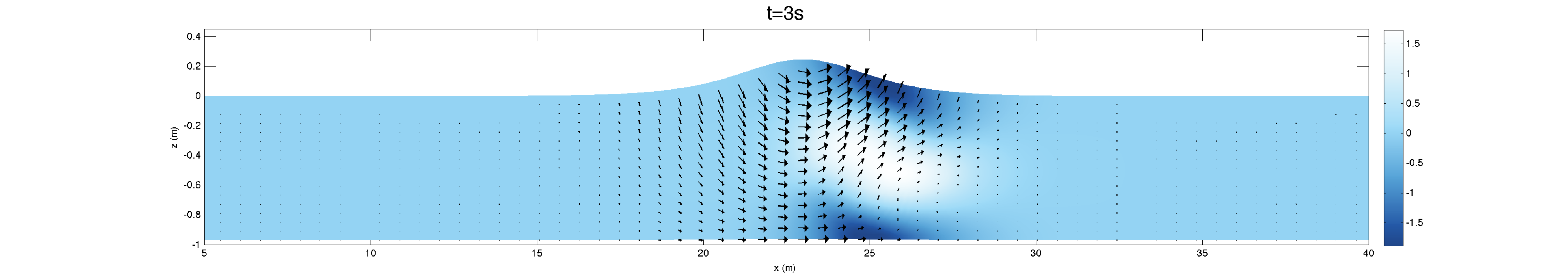}
\includegraphics[width=0.9\textwidth]{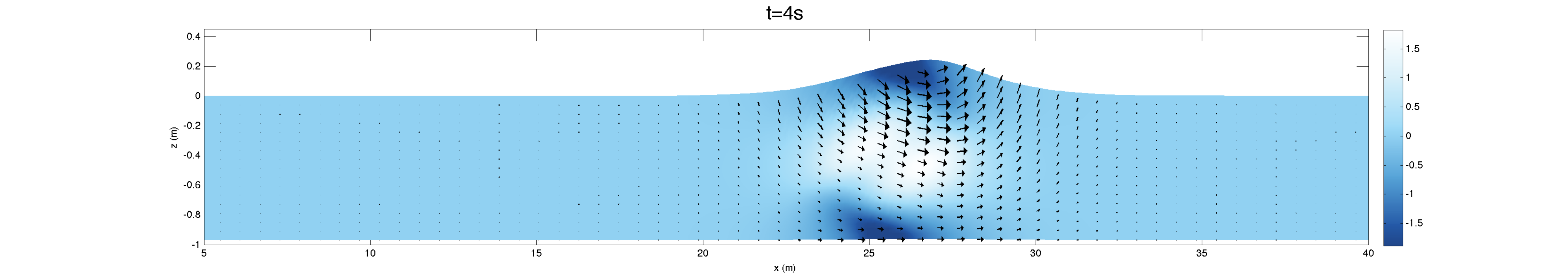}
\includegraphics[width=0.9\textwidth]{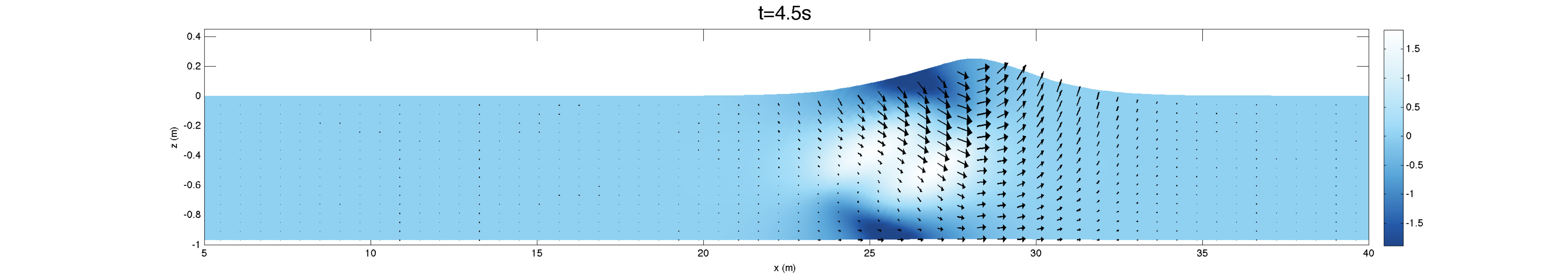}
\includegraphics[width=0.9\textwidth]{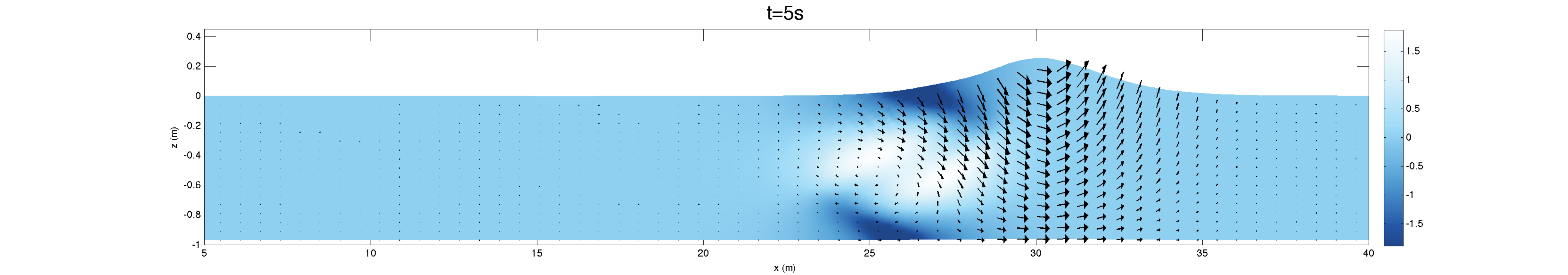}
\caption{A solitary wave arriving on a vorticity region: snapshot of the free surface profile at times $t=0, 2, 3, 4, 4.5$ and $5\,s$. The velocity fields structures are plotted with black arrows and the corresponding vorticity is shown according to the left colorbar.}
\label{fig2Dvelocity}
\end{figure}

\begin{figure}
\centering
\includegraphics[width=0.9\textwidth]{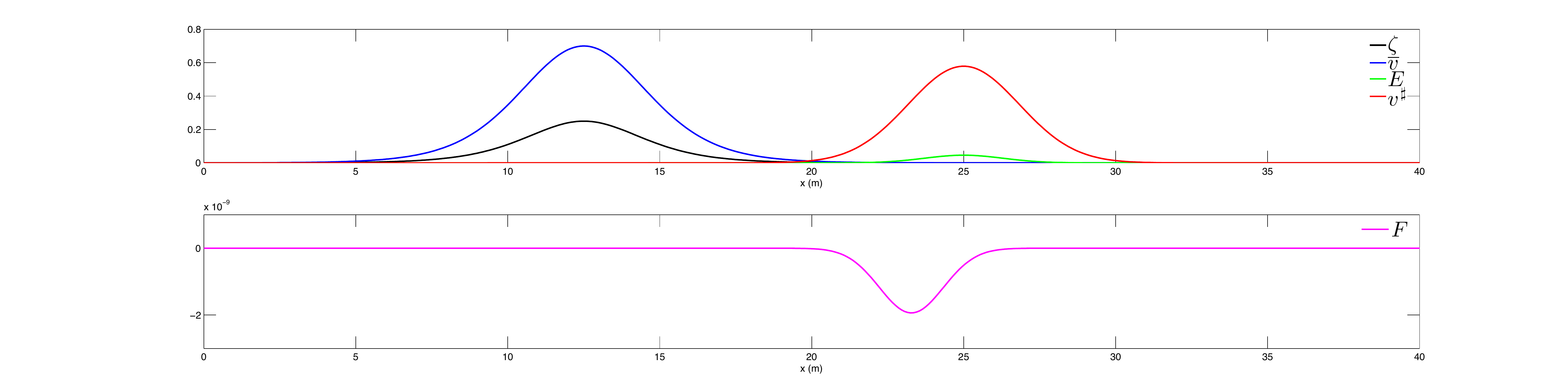}
\caption{Initial values for $\zeta$, $\ovv$ and $E$, $F$, $v^\sharp$
  for the configuration of \S \ref{sect2Dex1} with $H_0=1\,m$ and $\eps=0.25$.}
\label{figEFV}
\end{figure}

\subsection{Influence of vorticity on shoaling}\label{sect2Dex2}

We consider here the same configuration as in \S \ref{sect2Dex1} but
with the vorticity region now located in an area where the topography
varies. More precisely, instead of a flat bottom ($b=0$) we now
consider a bottom parametrized by $z=-H_0+b(x)$, with
$$
b(x)=\frac{H_0}{6}(1+\tanh\big(\frac{x-x_1}{\lambda}\big)).
$$
One could derive as in \S \ref{sect2Dex1}
explicit expressions for the initial values of $E$, $F$ and
$v^\sharp$, but they can also easily be numerically computed; they are
represented in Figure \ref{figEFV_shoal}. Note that in this
configuration, $F$ is still small (of order $10^{-3}$) but not as much
as in the previous configuration with a flat bottom. For this kind of
configurations, it is therefore suitable to work with the full
system \eqref{gn:vort:dim}. We show the corresponding results for the velocity fields on Figure \ref{fig2Dvelocity:topo}.
\begin{figure}
\centering
\includegraphics[width=0.9\textwidth]{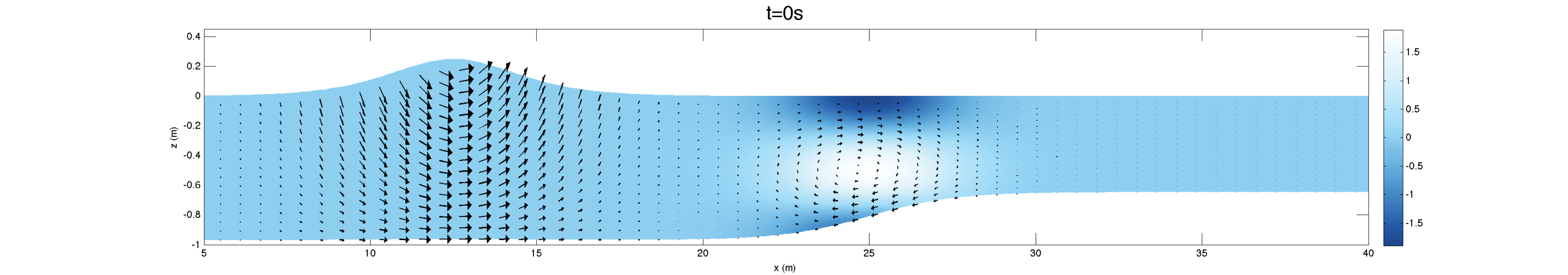}
\includegraphics[width=0.9\textwidth]{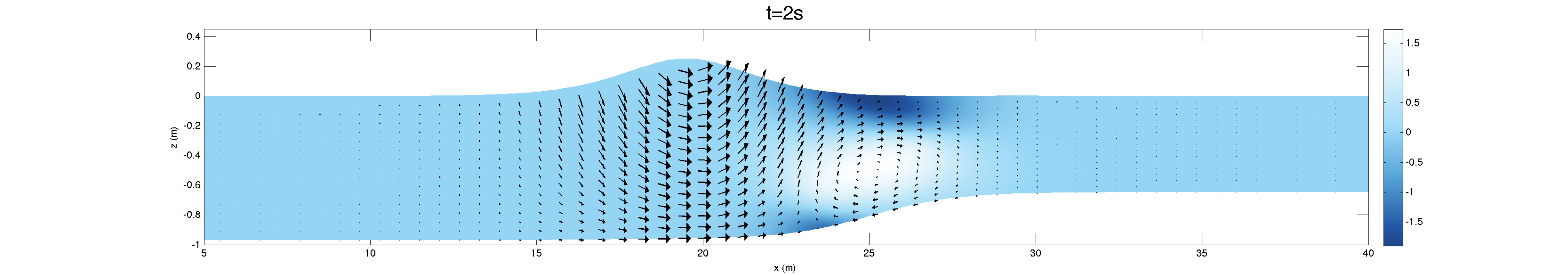}
\includegraphics[width=0.9\textwidth]{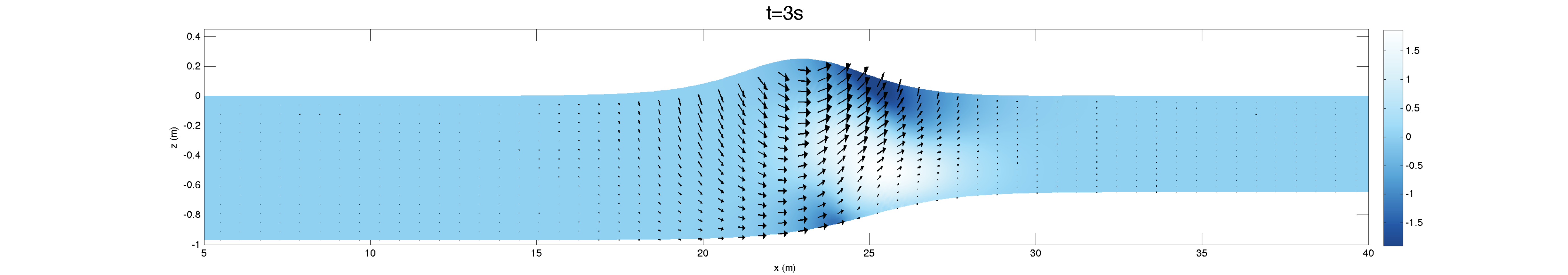}
\includegraphics[width=0.9\textwidth]{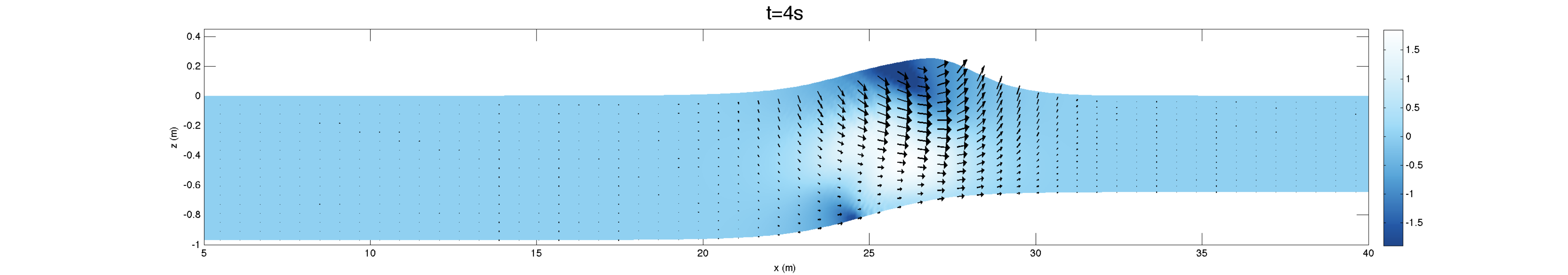}
\includegraphics[width=0.9\textwidth]{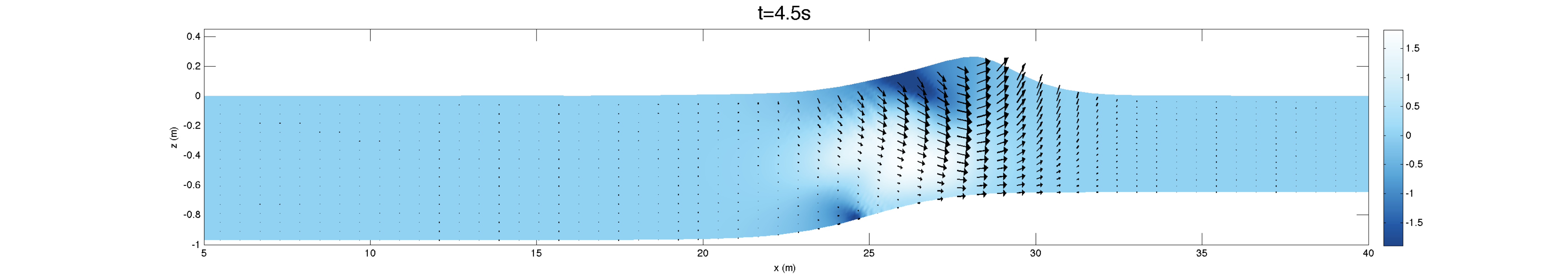}
\includegraphics[width=0.9\textwidth]{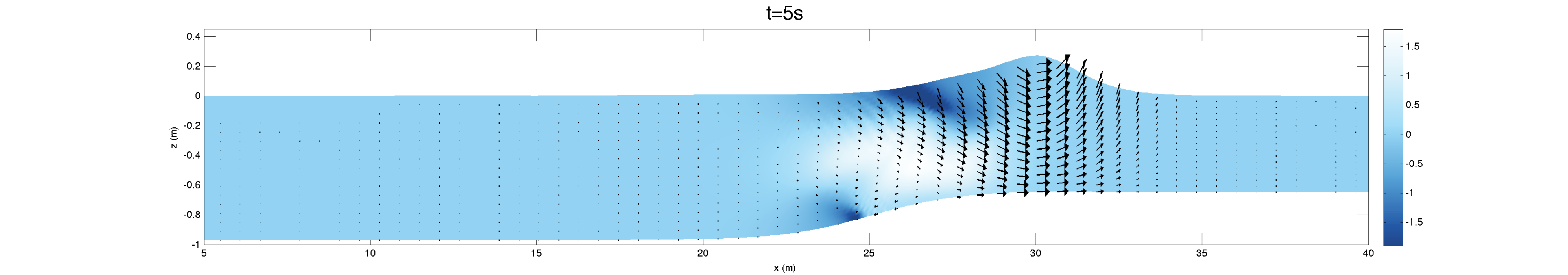}

\caption{Influence of vorticity on shoaling: snapshot of the free surface profile at times $t=0, 2, 3, 4, 4.5$ and $5\,s$. The velocity fields structures are plotted with black arrows and the corresponding vorticity is shown according to the left colorbar.}
\label{fig2Dvelocity:topo}
\end{figure}

\begin{figure}
\centering
\includegraphics[width=0.9\textwidth]{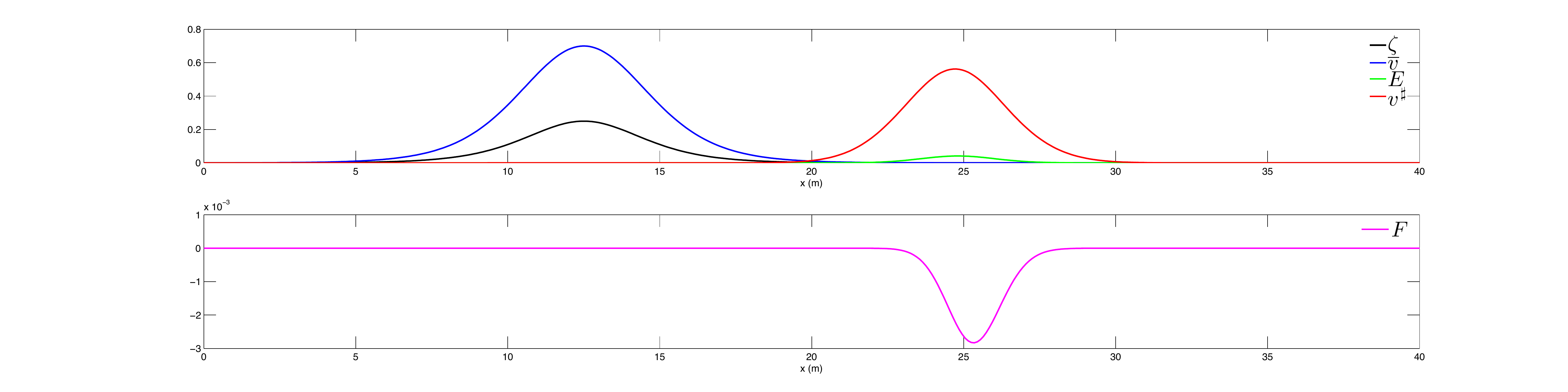}
\caption{Initial values for $\zeta$, $\ovv$ and $E$, $F$, $v^\sharp$
  for the configuration of \S \ref{sect2Dex2} with $H_0=1\,m$ and $\eps=0.25$.}
\label{figEFV_shoal}
\end{figure}

\section*{Acknowledgments}
D. Lannes acknowledges support from the ANR-13-BS01-0003-01 DYFICOLTI. F. Marche acknowledges support from the CNRS project LEFE-MANU SOLi. Both authors acknowledge support from the ANR-13-BS01-0009-01 BOND.






          %


%
%
%
%
%
%

          \end{document}